\documentclass[a4paper,UKenglish,cleveref, autoref, thm-restate]{lipics-v2021}

\pdfoutput=1 

\usepackage{complexity}
\usepackage{mathtools}
\usepackage{bm}
\usepackage{xspace,todonotes}
\usepackage{booktabs}

\newcommand{\fptpactime}{\ensuremath{\mathsf{FPT}\text{-}\mathsf{PAC_{time}}}}
\newcommand{\fptpac}{\ensuremath{\mathsf{FPT}\text{-}\mathsf{PAC}}}
\newcommand{\xppactime}{\mathsf{XP}\text{-}\mathsf{PAC_{time}}}
\newcommand{\xppac}{\ensuremath{\mathsf{XP}\text{-}\mathsf{PAC}}}

\newcommand{\fpt}{\mathrm{fpt}}

\newcommand{\xp}{\mathrm{xp}}

\newcommand{\bigoh}{\mathcal{O}}
\newcommand{\calC}{\mathcal{C}}
\newcommand{\calG}{\mathcal{G}}
\newcommand{\calL}{\mathcal{L}}
\newcommand{\calF}{\mathcal{F}}
\newcommand{\calI}{\mathcal{I}}

\newcommand{\calU}{\mathcal{U}}
\newcommand{\true}{\texttt{True}\xspace}
\newcommand{\false}{\texttt{False}\xspace}
\newcommand{\red}{\texttt{red}\xspace}
\newcommand{\blue}{\texttt{blue}\xspace}

\newcommand{\no}{\textsc{No}}
\newcommand{\N}{\ensuremath{\mathbb{N}}}
\newcommand{\Nat}{\N}
\newcommand{\supp}{\operatorname{supp}}

\newcommand{\CCcol}{\textup{\textsc{ConsCheck: $2$-Coloring}}\xspace}
\newcommand{\CCsplit}{\textup{\textsc{ConsCheck: Split Graph}}\xspace}
\newcommand{\CCmatch}{\textup{\textsc{ConsCheck: Matching}}\xspace}
\newcommand{\CCpath}{\textup{\textsc{ConsCheck: ($k$-)Path}}\xspace}
\newcommand{\CCEQC}{\textup{\textsc{ConsCheck: Edge Clique Cover}}\xspace}
\newcommand{\CCIS}{\textup{\textsc{ConsCheck: Independent Set}[degree]}\xspace}

\newcommand{\CCDS}{\textup{\textsc{ConsCheck: Dominating Set}[degree]}\xspace}
\newcommand{\SATpr}{\textup{\textsc{Satisfability}}\xspace}
\newcommand{\SCpr}{\textup{\textsc{Set Cover}}\xspace}

\usepackage{tikz}
\usetikzlibrary{positioning,backgrounds,patterns,calc,decorations.pathreplacing,calligraphy}

\def\problem#1#2#3{
\vspace{0.2cm}\noindent\vspace{0.2cm}
\begin{tabular}{l}
\hline
\textsc{#1}\\
\hline
\textbf{Input:} #2\\
\textbf{Output:} #3\\
\hline
\end{tabular}
}

\def\problempar#1#2#3#4{
\vspace{0.2cm}\noindent\vspace{0.2cm}
\begin{tabular}{l}
\hline
\textsc{#1}\\
\hline
\textbf{Input:} #2\\
\textbf{Parameter:} #3\\
\textbf{Output:} #4\\
\hline
\end{tabular}
}

\tikzset{
  circ/.style = {circle,draw,fill=white,inner sep=1.3pt}
}

\hideLIPIcs  

\title{Consistency-Checking Problems: A Gateway to Parameterized Sample Complexity}

\titlerunning{Consistency-Checking Problems: A Gateway to Parameterized Sample Complexity}

\author{Robert Ganian}{Technische Universit\"{a}t Wien, Vienna, Austria \and \url{https://www.ac.tuwien.ac.at/people/rganian/}}{rganian@gmail.com}{https://orcid.org/my-orcid?orcid=0009-0002-3012-7240}{Austrian Science Fund (FWF) [Y1329].}

\author{Liana Khazaliya}{Technische Universit\"{a}t Wien, Vienna, Austria \and \url{https://www.ac.tuwien.ac.at/people/lkhazaliya/}}{lkhazaliya@ac.tuwien.ac.at}{https://orcid.org/my-orcid?orcid=0009-0002-3012-7240}{Vienna Science and Technology Fund (WWTF) [10.47379/ICT22029]; Austrian Science Fund (FWF) [Y1329]; European Union's Horizon 2020 COFUND programme [LogiCS@TUWien, grant agreement No.\ 101034440].}

\author{Kirill Simonov}{Hasso Plattner Institute, Universität Potsdam, Germany \and \url{https://hpi.de/friedrich/people/kirill-simonov.html}}{kirillsimonov@gmail.com}{https://orcid.org/0000-0001-9436-7310}{DFG Research Group ADYN via grant DFG 411362735.}

\authorrunning{R.\ Ganian, L.\ Khazaliya, and K.\ Simonov}

\Copyright{Robert Ganian, Liana Khazaliya, and Kirill Simonov} 


\nolinenumbers

\EventEditors{}
\EventNoEds{2}
\EventLongTitle{18th International Symposium on Parameterized and Exact Computation (IPEC 2023)}
\EventShortTitle{IPEC 2023}
\EventAcronym{IPEC}
\EventYear{2023}
\EventDate{}
\EventLocation{Amsterdam, Netherlands}
\EventLogo{}
\SeriesVolume{}
\ArticleNo{}

\begin{document}

\maketitle

\begin{abstract}
Recently, Brand, Ganian and Simonov introduced a parameterized refinement of the classical PAC-learning sample complexity framework. A crucial outcome of their investigation is that for a very wide range of learning problems, there is a direct and provable correspondence between fixed-parameter PAC-learnability (in the sample complexity setting) and the fixed-parameter tractability of a corresponding ``consistency checking'' search problem (in the setting of computational complexity). The latter can be seen as generalizations of classical search problems where instead of receiving a single instance, one receives multiple yes- and no-examples and is tasked with finding a solution which is consistent with the provided examples.

Apart from a few initial results, consistency checking problems are almost entirely unexplored from a parameterized complexity perspective. In this article, we provide an overview of these problems and their connection to parameterized sample complexity, with the primary aim of facilitating further research in this direction. Afterwards, we establish the fixed-parameter (in)-tractability for some of the arguably most natural consistency checking problems on graphs, and show that their complexity-theoretic behavior is surprisingly very different from that of classical decision problems. Our new results cover consistency checking variants of problems as diverse as \textsc{($k$-)Path}, \textsc{Matching}, \textsc{2-Coloring}, \textsc{Independent Set} and \textsc{Dominating Set}, among others.
\end{abstract}

\section{Introduction}

While the notion of time complexity is universally applicable and well studied across the whole spectrum of theoretical computer science, on its own it cannot capture the performance of the kinds of algorithms typically studied in the context of machine learning: \textbf{learning algorithms}. That is the domain of sample complexity, and here we will focus on the notion of (efficient) \emph{PAC learning}~\cite{Valiant84,KVazirani94}---arguably the most classical, fundamental and widely known sample complexity framework. An important trait of PAC learning is that while it is built on different principles than time complexity, the two frameworks are connected in a way which allows us to translate intractability and tractability results from one domain to another. It is precisely this connection that gave rise to famous lower bounds in the PAC learning setting, such as the inability to efficiently and properly learn 3-term DNF and $3$-clause CNF formulas~\cite{PittV88,AlekhnovichBFKP08} under the assumption that $\P\neq\NP$, and \emph{consistency checking problems} form the pillar of this connection.

Given the success of parameterized complexity as a concept generalizing classical time complexity analysis, it would seem natural to ask whether its principles can also be used to obtain a deeper understanding of efficient PAC-learnability. Brand, Ganian and Simonov~\cite{BrandGanianSimonov23} very recently introduced the foundations for a parameterized theory of PAC learning, which crucially also includes a bridge to parameterized complexity theory in the usual time complexity setting. The primary goal of this article is to show how the parameterized complexity paradigm can be used to draw new boundaries of tractability in the PAC learning domain, and to provide the parameterized algorithms community with an understanding of the parameterized consistency checking problems which allow us to travel between the sample and time complexity settings in the parameterized regime. We showcase the tools that can be used to deal with parameterized consistency checking problems and the obstacles that await there in the domain of graph problems, where we obtain new algorithmic upper and lower bounds for consistency checking variants of multiple natural problems on graphs.

\smallskip
\noindent \textbf{A Gentle Introduction to PAC Learning.}\quad
It will be useful to set the stage with a high-level and informal example of the setting in which PAC learning operates\footnote{Formal definitions are provided in \Cref{sec:prelims}.}. Let us imagine we would like to ``learn'' a way of labeling points in a plane as either ``good'' or ``bad'', knowing that the good points are precisely those contained in some unknown axis-parallel rectangle $R$ in the plane. A learning algorithm in the PAC regime would be allowed to ask for a set of correctly labeled sample points, each of which would be drawn from some unknown distribution $D$, and would attempt to use these to ``learn'' $R$ (so that it can use it to label any point that it looks at, even those which were not given as samples). This mental experiment is useful since it immediately clarifies that
\begin{itemize}
\item there is some probability that a PAC learning algorithm completely fails, since the samples we receive could be non-representative (\emph{for instance, there is a non-zero probability that even if $D$ is uniform and $R$ is small, the sample points could all be drawn from inside $R$}), and
\item even if a PAC learning algorithm intuitively ``works correctly'', it is essentially guaranteed that it will not classify some samples (i.e., the sample points) correctly (\emph{for instance, there could be points that lie close to the exact boundary of $R$ which are unlikely to be drawn as samples based on $D$, making it impossible to obtain the exact position of $R$}).
\end{itemize}
Given these natural limitations, we can informally explain what it means for a learning problem to be \emph{efficiently PAC-learnable}: it admits an algorithm which 
\begin{enumerate}
\item takes as input a sample size $n$, a confidence measure $\delta$ and an accuracy measure $\varepsilon$, 
\item runs in time $(n+\frac{1}{\delta}+\frac{1}{\varepsilon})^{\bigoh(1)}$ and asks for $(n+\frac{1}{\delta}+\frac{1}{\varepsilon})^{\bigoh(1)}$ samples, and then 
\item outputs something which will, with probability at least $1-\delta$, ``work correctly'' in almost all cases (measured by $\varepsilon$).
\end{enumerate}

It needs to be clarified that beyond the study of efficient PAC learnability, a substantial amount of fundamental work in the PAC learning direction has also been carried out on whether a problem is PAC learnable at all~\cite{BlumerEHW89,Hanneke16,AgarwalABLU21}, on the distinction between so-called proper and improper learning~\cite{KVazirani94,BousquetHMZ20}, and on many other aspects and considerations that lie outside of the scope of this paper. Here, our focus lies on how the gap between efficient and ``non-efficient'' PAC-learnability of learning problems can be bridged by the parameterized PAC learning framework of Brand, Ganian and Simonov~\cite{BrandGanianSimonov23}, and the associated study of consistency checking problems. 

To illustrate how parameterized complexity can be used here, let us turn to a different example of a simple learning problem that is based on the idea of representing cyber-attacks as graphs proposed, e.g., by Sheyner and Wing~\cite{SheynerW03,Wing06}. Assume we have a network consisting of $n$ nodes which is protected by $k$ hidden defense nodes. A cyberattack on this network can be represented as a set of edges over the $n$ nodes, and is evaluated as successful if and only if an edge in that attack is not incident to any defense node (i.e., an attack fails if and only if the defense nodes form a vertex cover of the attack edges). Individual samples represent attacks made on the network, and the learning task is to identify all the defense nodes. This problem corresponds to \textsc{Vertex Cover Learning}~\cite{DamaschkeM10},
which Brand, Ganian and Simonov showed to admit a PAC learning algorithm which requires polynomially many samples but time $2^k\cdot (n+\frac{1}{\delta}+\frac{1}{\varepsilon})^{\bigoh(1)}$ where $k$ is the size of the sought-after vertex cover~\cite{BrandGanianSimonov23}. This is a prototypical representative of the class \fptpactime. We remark that in the context of PAC learning, one explicitly distinguishes between the time required by the learning algorithm and the number of samples it uses, as the latter may in some contexts be much more difficult to obtain.
A picture of the parameterized complexity landscape above efficient PAC learnability is provided later together with the formal definitions (see Figure~\ref{fig:classes}).

Crucially, whenever we are dealing with a learning problem $\mathcal{P}_\texttt{learn}$ where the size of the hypothesis space (i.e., the number of ``possible outputs'') is upper-bounded by a certain function (see \Cref{thm:red-learn-to-dec}), the parameterized sample complexity of $\mathcal{P}_\texttt{learn}$ can be directly and formally linked to the parameterized time complexity of the consistency checking variant $\mathcal{P}_\texttt{cons}$ of the same problem~\cite{BrandGanianSimonov23}, where the task is to compute a ``solution'' (a hypothesis) which is consistent with a provided set of positive and negative examples. This motivates the systematic study of parameterized consistency checking problems, an area which has up to now remained almost entirely unexplored from the perspective of fixed-parameter (in-)tractability. 

\smallskip
\noindent \textbf{The Parameterized Complexity of Consistency Checking on Graphs.}\quad
A few initial examples of parameterized consistency checking problems have been solved by the theory-building work of Brand, Ganian and Simonov~\cite{BrandGanianSimonov23}; in particular, they showed that consistency checking for vertex-deletion problems where the base class $\mathcal{H}$ can be characterized by a finite set of forbidden induced subgraphs is fixed-parameter tractable (which implies the aforementioned fact that \textsc{Vertex Cover Learning} is in \fptpactime), but no analogous result can be obtained for all classes $\mathcal{H}$ characterized by a finite set of forbidden minors unless $\FPT\neq \W[1]$. 

In this article, we expand on these results by establishing the fixed-parameter (in-)tractability of consistency checking for several other classical graph problems whose decision versions are well-known to the parameterized complexity community. 
The aim here is to showcase how parameterized upper- and lower-bound techniques fare when dealing with these new kinds of problems. 

It is important to note that the tractability of consistency checking requires the tractability of the corresponding decision/search problem (as the latter can be seen as a special case of consistency checking), but the former can be much more algorithmically challenging than the latter: many trivial decision problems become computationally intractable in the consistency checking regime. We begin by illustrating this behavior on the classical $2$-\textsc{Coloring} problem, i.e., the task of partitioning the vertices of the graph into two independent sets. We show that while consistency checking for $2$-\textsc{Coloring} is intractable (and hence a $2$-coloring is not efficiently PAC-learnable), consistency checking for \textsc{Split Graph}, i.e., the task of partitioning the vertices into an independent set and a clique, is polynomial-time tractable.

Moving on to parameterized problems, we begin by considering three classical edge search problems, notably \textsc{Matching}, ($k$-)\textsc{Path} and \textsc{Edge Clique Cover}. 
In the classical decision or search settings, the first problem is polynomial-time solvable while the latter two admit well-known fixed-parameter algorithms.
Interestingly, we show that consistency checking for the former two problems is \W[2]-hard\footnote{More precisely, a fixed-parameter algorithm for either of these problems would imply \FPT=\W[1] (see Section~\ref{sec:edgesearch}).}, but is fixed-parameter tractable for the third, i.e., \textsc{Edge Clique Cover}. 

Next, we turn our attention to the behavior of two classical vertex search problems, specifically \textsc{Independent Set} and \textsc{Dominating Set}. While both problems are fixed-parameter intractable already in the classical search regime, here we examine their behavior on bounded-degree graphs (where they are well-known to be fixed-parameter tractable). Again, the consistency checking variants of these problems on bounded-degree graphs exhibit a surprising complexity-theoretic behavior: \textsc{Dominating Set} is \FPT, but \textsc{Independent Set} is \W[2]-hard even on bounded-degree graphs. 

As the final contribution of the paper, we show that most of the aforementioned consistency checking lower bounds can be overcome if one additionally parameterizes by the number of negative samples. In particular, we obtain fixed-parameter consistency checking algorithms for $2$-\textsc{Coloring}, \textsc{Matching} and ($k$-)\textsc{Path}  when we additionally assume that the number of negative samples is upper-bounded by the parameter. On the other hand, \textsc{Independent Set} remains fixed-parameter intractable (at least \W[1]-hard) even under this additional restriction. As our final result, we show that \textsc{Independent Set} becomes fixed-parameter tractable if we instead consider the total number of samples (i.e., both positive and negative) as an additional parameter.
The proofs of these results are more involved than those mentioned in the previous paragraphs and rely on auxiliary graph constructions in combination with color coding. We remark that the parameterization by the number of negative samples in the consistency checking regime could be translated into a corresponding parameterization of the distribution in the PAC learning framework. A summary of our individual results for consistency checking problems is provided in Table~\ref{tab:problems}.

\begin{table*}[t]
  \begin{center}
   \scalebox{0.93}{
    \begin{tabular}{@{}l@{\quad}c@{\quad}c@{\quad}c@{\quad}c@{}}\toprule

Problem & Decision/Search & Consistency Checking & Consistency Checking[samples] \\
\hline
\textsc{$2$-Coloring} & \P & \NP-hard (Thm.~\ref{thm:2colNP}) & \FPT\ (Thm.~\ref{thm:ColFPT}) \\
\textsc{Split Graph} & \P & \P\ (Thm.~\ref{thm:splitcons}) & --- \\
\textsc{Matching} & \P & \W[2]-hard (Thm.~\ref{thm:CCmatchhard}) & \FPT\ (Thm.~\ref{thm:CCmatchfpt}) \\
\textsc{($k$)-Path} & \FPT & \W[2]-hard (Thm.~\ref{thm:CCpathhard}) & \FPT\ (Thm.~\ref{thm:CCpathfpt}) \\
\textsc{Edge Clique Cover} & \FPT & \FPT\ (Thm.~\ref{thm:CCEDS}) & --- \\
\textsc{Independent Set}[degree] & \FPT & \W[2]-hard (Thm.~\ref{thm:ISWtwo}) & \W[1]-hard$^\star$ (Thm.~\ref{thm:ISWone},~\ref{thm:CCISfpt})\\
\textsc{Dominating Set}[degree] & \FPT & \FPT\ (Thm.~\ref{thm:DSFPT}) & ---\\
     \bottomrule
    \end{tabular}
      }
  \end{center}

  \caption{\small An overview of the concrete results obtained for consistency checking problems in this article, where the columns provide a comparison between the complexity of the decision/search variant, the consistency checking variant, and the consistency checking variant where the number of negative samples is taken as an additional parameter. Problems marked with ``[degree]'' are considered over bounded-degree input graphs/samples, and the ``$^\star$'' marks that the problem becomes fixed-parameter tractable when additionally parameterized by the total number of samples. The lower bounds stated in the table are simplified; the precise formal statements are provided in the appropriate theorems.} 
  \label{tab:problems}
\end{table*}

\smallskip
\noindent \textbf{Related Work.}\quad
The connection between parameterized learning problems and parameterized consistency checking was also hinted at in previous works that studied the (parameterized) sample complexity of learning juntas~\cite{ArvindKL09} or learning first-order logic~\cite{BergeremGR22}.
Moreover, the problem of computing optimal decision trees, which has received a significant amount of recent attention~\cite{OrdyniakSzeider21,EibenOrdyniakPaesaniSzeider23}, can also be seen as a consistency checking problem where the sought-after solution is a decision tree.

\section{Preliminaries}
\label{sec:prelims}

We assume familiarity with basic graph terminology~\cite{Diestel17} and parameterized complexity theory~\cite{DBLP:books/sp/CyganFKLMPPS15}. We use $[t]$ to denote the set $\{1,\dots,t\}$. For brevity, we will denote sets of tuples of the form $\{(\alpha_1,\beta_1),\dots,(\alpha_t,\beta_t)\}$ as $(\alpha_i,\beta_i)_{i\in [t]}$, and the set of two-element subsets of a set $Z$ as $Z\choose 2$.
As basic notation and terminology, we set $\{0,1\}^\ast = \bigcup_{m\in\N} \{0,1\}^m$.
A \emph{distribution} on $\{0,1\}^n$ is a mapping $\mathcal{D}_n: \{0,1\}^n \rightarrow [0,1]$ such that $\sum_{x \in \{0,1\}^n} \mathcal{D}_n(x) = 1$, and the \emph{support} of $\mathcal{D}_n$ is the set $\supp\mathcal{D}_n = \{ x \mid \mathcal{D}_n(x) > 0\}$.

\subsection{Consistency Checking}
While the original motivation for consistency checking problems originates from specific applications in PAC learning, one can define a consistency checking version of an arbitrary \emph{search problem}.

In a search problem, we are given an instance $I \in \{0,1\}^\ast$, and the task is to find a solution $S \in \{0,1\}^\ast$, where the solution is verified by a predicate $\phi(\cdot, \cdot)$, so that $\phi(I, S)$ is true if and only if $S$ is a solution to $I$. Since our focus here will lie on problems which are in \NP, the predicate $\phi(\cdot, \cdot)$ will in all cases be polynomial-time computable. In the context of graph problems, $I$ will typically be a graph (possibly with some auxiliary information such as edge weights or the bound on solution size), and $S$ could be a set of vertices, a set of edges, a partitioning of the vertex set, etc. For example, in the search version of the \textsc{Vertex Cover} problem the input is a graph $G$ together with a bound $k$ on the size of the target vertex cover, potential solutions are subsets of $V(G)$, and a subset $S$ is a solution if and only if the size of $S$ is $k$ and $S$ covers all edges of the graph $G$. One can then write the verifying predicate as
\begin{equation*}
\phi\left((G, k), S\right) = \left(S \subset V(G)\right) \land (|S| = k) \land (\forall \{u, v\} \in E(G), \{u, v\} \cap S \ne \emptyset).
\end{equation*}

For a search problem $\mathcal{P}$, we define the corresponding consistency checking problem $\mathcal{P}_{\text{cons}}$ as follows. Instead of receiving a single instance $I \in \{0,1\}^\ast$ as input, we receive a set of labeled samples $\mathcal{I}=\{(I_1,\lambda_1),(I_2,\lambda_2), \ldots, (I_t,\lambda_t)\}$ where each $I_i, i\in [t],$ is an element of $\{0,1\}^\ast$ and $\lambda_i\in \{0,1\}$. The task is to compute a (\emph{consistent}) \emph{solution} $S \subset \{0,1\}^\ast$ such that $\phi(I_i, S)$ holds if and only if $\lambda_i=1$, for each $i \in [t]$, or to correctly determine that no such solution exists.

In the example of \textsc{Vertex Cover}, for each $i \in [t]$, the instance is the pair $(G_i, k_i)$, so that the target solution has to be a vertex subset\footnote{The property of being a subset is given by the implicit encoding in $ \{0,1\}^\ast$, e.g., vertices in all $V(G_i)$ and $S$ are indexed by integers, and is defined in the same way across all instances. We thus say that $S$ could be a subset of all $V(G_i)$ even though, formally speaking, these are disjoint sets.} of $V(G_i)$, of size $k_i$, and it has to cover all edges of $G_i$, for each $i \in [t]$. Since vertices in all $G_i$'s and $S$ are implicitly associated with their respective counterparts in the other graphs, we can instead treat the graphs $G_i$ as defined over the same vertex set. Also, for instances $i \in [t]$ where $\lambda_i = 1$, if their values of $k_i$ mismatch, then there is clearly no solution; and for those $i \in [t]$ with $\lambda_i = 0$, if the value $k_i$ does not match the respective value of a positive sample, then the condition for $\lambda_i$ is always satisfied. Therefore, we can equivalently reformulate the consistency checking version of \textsc{Vertex Cover} as follows: Given the vertex set $V$, a number $k$, and a sequence of labeled edge sets $(E_1, \lambda_1), \ldots, (E_t, \lambda_t)$,  over $V$, is there a subset $S \subset V$ of size exactly $k$, so that $S$ covers all edges of $E_i$ if and only if $\lambda_i = 1$, for each $i \in [t]$?

One can immediately observe that the polynomial-time tractability of a search problem is a prerequisite for the polynomial-time tractability of the corresponding consistency checking problem. At the same time, the exact definition of the search problem (and in particular the solution $S$) can have a significant impact on the complexity of the consistency checking problem. We remark that there are two possible ways one can parameterize a consistency checking problem: one either uses the parameter to restrict the sought-after solution $S$, or the input $\mathcal{I}$. Each of these approaches can be tied to a parameterization of the corresponding PAC learning problem (see Subsection~\ref{sub:paraPAC}).

Formally, we say that $(\mathcal{P}_{\text{cons}}, \kappa, \lambda)$ is a parameterized consistency checking problem, where $\mathcal{P}_{\text{cons}}$ is a consistency checking problem, $\kappa$ maps solutions $S \in \{0, 1\}^\ast$ to natural numbers, and $\lambda$ maps lists of labeled instances $((I_1, \lambda_1), \ldots, (I_t, \lambda_t))$, $I_i \in \{0, 1\}^*$, $\lambda_i \in \{0, 1\}$, to natural numbers.  The input is then a list of labeled instances $\mathcal{L} = ((I_1, \lambda_1), \ldots, (I_t, \lambda_t))$ together with parameters $k, \ell$, such that $\ell = \lambda(\mathcal{L})$, and the task is to find a consistent solution $S$ with $\kappa(S) = k$.
For example, $k$ could be a size bound on the targeted solution, and $\ell$ could be the maximum degree in any of the given graphs or the number of instances with $\lambda_i = 0$.

\subsection{PAC-Learning}
\label{sub:pac}
The remainder of this section is dedicated to a more formal introduction of the foundations of parameterized PAC learning theory and its connection to parameterized consistency checking problems. We note that while the content of the following subsections is important to establish the implications and corollaries of the results obtained in the article, readers who are interested solely in the obtained complexity-theoretic upper and lower bounds for consistency checking problems can safely skip them and proceed directly to Section~\ref{sec:partition}.

To make the connection between consistency checking problems and parameterized sample complexity clear, we first recall the formalization of the classical theory of PAC learning~\cite{Valiant84,MLbook}.

\begin{definition}
A \emph{concept} is an arbitrary Boolean function $c: \{0,1\}^n \rightarrow \{0,1\}$.
An assignment $x\in \{0,1\}^n$ is called a \emph{positive sample} for $c$ if $c(x) = 1$, and a \emph{negative sample} otherwise. A \emph{concept class} $\mathcal C$ is a set of concepts.
For every $m\in \Nat$, we write $\mathcal C_m = \mathcal C \cap \mathcal{B}_m$, where $\mathcal{B}_m$ is the set of all $m$-ary Boolean functions.

\end{definition}

\begin{definition}
Let $\mathcal{C}$ be a concept class. A surjective mapping $\rho: \{0,1\}^\ast \rightarrow \mathcal{C}$ is called a \emph{representation scheme} of $\mathcal{C}$. 

We call each $r$ with $\rho(r) = c$ a \emph{representation} of concept $c$.

\end{definition}

\begin{definition}
A \emph{learning problem} is a pair $(\mathcal{C}, \rho)$, where $\mathcal{C}$ is a concept class and $\rho$ is a representation scheme for $\mathcal{C}$.
\end{definition}

\begin{definition} \label{def:learning-alg}
A \emph{learning algorithm} for a learning problem $(\mathcal{C},\rho)$ is a randomized algorithm such that:
\begin{enumerate}
\item It obtains the values $n,\varepsilon,\delta$ as inputs, where $n$ is an integer and $0< \varepsilon,\delta \leq 1$ are rational numbers.

\item It has access to a hidden representation $r^\ast$ of some concept $c^\ast = \rho(r^\ast)$ and a hidden distribution $\mathcal{D}_n$ on $\{0,1\}^n$ through an oracle that returns \emph{labeled samples} $(x,c^\ast(x))$, where $x \in \{0,1\}^n$ is drawn at random from $\mathcal{D}_n$.
\item The output of the algorithm is a representation of some concept, called its \emph{hypothesis}.
\end{enumerate}
\end{definition}

When dealing with individual instances of a learning problem, we will use $s=|r^\ast|$ to denote the size of the hidden representation.

\begin{definition}
Let $\mathcal{A}$ be a learning algorithm. Fix a hidden hypothesis $c^\ast$ and a distribution on $\{0,1\}^n$.
Let $h$ be a hypothesis output by $\mathcal{A}$ and $c = \rho(h)$ be the concept $h$ represents.
We define 
\[
\operatorname{err}_h = \mathbb{P}_{x \sim \mathcal{D}_n}(c(x) \neq c^\ast(x))
\]
as the probability of the hypothesis and the hidden concept disagreeing on a sample drawn from $\mathcal{D}_n$,  the so-called \emph{generalization error} of $h$ under $\mathcal{D}_n$.

The algorithm $\mathcal{A}$ is called \emph{probably approximately correct (PAC)} if it outputs a hypo\-thesis $h$ such that $\operatorname{err}_h \leq \varepsilon$ with probability at least $1-\delta.$ 
\end{definition}

Usually, learning problems in this framework are regarded as tractable if they are PAC-learnable within polynomial time bounds. More precisely, we say that a learning problem $L$ is \emph{efficiently} PAC-learnable if there is a PAC algorithm for $L$ that runs in time polynomial in $n,s,1/\varepsilon$ and $1/\delta$.

Consider now a classical search problem $\mathcal{P}$ and its consistency checking version $\mathcal{P}_{\text{cons}}$. One can naturally define the corresponding learning problem $\mathcal{P}_{\text{learn}}$: For a solution $S \in \{0, 1\}^\ast$, let $\phi(\cdot, S)$ be a concept and $S$ its representation; this describes the concept class and its representation scheme. Going back to the \textsc{Vertex Cover} example, for each graph size $N$, the concepts are represented by subsets of $[N]$ (encoded in binary). For a subset $S \subset [N]$, the respective concept $c_S$ is a binary function that, given the encoding of an instance $E$, returns $1$ if and only if $S$ is a vertex cover of $G = ([N], E)$ of size $k$, where $[N]$ is treated as the respective ``ground'' vertex set of size $N$. A PAC-learning algorithm for this problem is thus given a vertex set $V = [N]$, an integer $k$, and an oracle that will produce a sequence of samples $(E_1, \lambda_1), \ldots, (E_t, \lambda_t)$, where the instances $E_i$ are drawn from a hidden distribution $\mathcal{D}$. With probability at least $(1 - \delta)$, the algorithm has to return a subset $S \subset [N]$ that is consistent with an instance sampled from $\mathcal{D}$ with probability at least $(1 - \varepsilon)$. In fact, for \textsc{Vertex Cover} and many other problems, it is sufficient to return a hypothesis that is consistent only with the seen samples $(E_i,\lambda_i)$, $i\in [t]$; this is formalized in the next subsection.

Naturally, we do not expect the learning version of \textsc{Vertex Cover} to be efficiently PAC-learnable, as even finding a vertex cover of a certain size in a single instance is \NP-hard.
This motivates the introduction of parameters into the framework, which is presented next.
We also recall the complexity reductions between (parameterized) consistency checking problem and its respective (parameterized) learning problem, which in particular allows to formally transfer the hardness results such as \NP-hardness above.

\smallskip
\noindent \emph{Remark.}\quad A more general definition of learning problems is sometimes considered in the literature, where the output of a learning algorithm need not necessarily be from the same concept class $\mathcal{C}$ (e.g., it can be a sub- or a super-class of $\mathcal{C}$). This is usually called \emph{improper learning}, as opposed to the classical setting of \emph{proper learning} defined above and considered in this article.

\subsection{Parameterized PAC-Learning}
\label{sub:paraPAC}

We now define parameterized learning problems and recall the connection to the consistency checking problems, as given by the framework of Brand, Ganian, and Simonov~\cite{BrandGanianSimonov23}. For brevity, we omit some of the less important technical details; interested readers can find the full technical exposition in the full description of the framework~\cite{BrandGanianSimonov23}

First we note that in parameterized PAC-learning, both the hidden concept and the hidden distribution can be parameterized, which is formally represented in the next definitions. We call a function $\kappa$ from representations in $\{0, 1\}^\ast$ to natural numbers \emph{parameterization of representations}, and a function $\lambda$ assigning a natural number to every distribution on $\{0, 1\}^n$ for each $n$ \emph{parameterization of distributions}.

\begin{definition}[Parameterized Learning Problems]
A \emph{parameterized learning problem} is a learning problem $(\mathcal{C},\rho)$ together with a pair $(\kappa, \lambda)$, called its \emph{parameters}, where $\kappa$ is a parameterization of representations and $\lambda$ is a parameterization of distributions.

\end{definition}

\begin{definition}[Parameterized Learning Algorithm]
A \emph{parameterized learning algorithm} for a parameterized learning problem $(\mathcal{C},\rho,\kappa,\lambda)$ is a learning algorithm for $(\mathcal{C},\rho)$ in the sense of Definition \ref{def:learning-alg}.
In addition to $n, \varepsilon, \delta$, a parameterized learning algorithm obtains two inputs $k$ and $\ell$, which are promised to satisfy $k = \kappa(r^\ast)$ as well as $\ell= \lambda(\mathcal{D}_n)$,
and the algorithm is required to always output a hypothesis $h$ satisfying $\kappa(h) \le k$.
\end{definition}

Let $\poly(\cdot)$ denote the set of functions that can be bounded by non-decreasing polynomial functions in their arguments.
Furthermore, $\fpt(x_1,\ldots,x_t;k_1,\ldots,k_t)$ and $\xp(x_1,\ldots,x_t;k_1,\ldots,$ $k_t)$ denote those functions bounded by $f(k_1,\ldots,k_t)\cdot p(x_1,\ldots,x_t)$ and $p(x_1,\ldots,x_t)^{f(k_1,\ldots,k_t)}$, respectively, for any non-decreasing computable function $f$ in $k_1,\ldots,k_t$ and $p \in \poly(x_1,\ldots,x_t)$.

\begin{definition}[$(T,S)$-PAC Learnability] \label{def:learnable}
Let $T(n,s,1/\varepsilon,1/\delta,k,\ell),S(n,s,1/\varepsilon,1/\delta,k,\ell)$ be any two functions taking on integer values, and non-decreasing in all of their arguments.

A parameterized learning problem $\mathcal{L} = (\mathcal{C},\rho, \{\mathcal{R}_k\}_{k\in\N},\lambda)$ is \emph{$(T,S)$-PAC learnable} if there is a PAC learning algorithm for $\mathcal{L}$ that runs in time $\bigoh(T(n,s,1/\varepsilon,1/\delta,k,\ell))$ and queries the oracle at most $\bigoh(S(n,s,1/\varepsilon,1/\delta,k,\ell))$ times. 

We denote the set of parameterized learning problems that are $(T,S)$-PAC learnable by $\mathsf{PAC}[T,S]$. This is extended to sets of functions $\mathbf{S},\mathbf{T}$ through setting $\mathsf{PAC}[T,S] = \bigcup_{S \in \mathbf{S},\\T \in \mathbf{T}} \mathsf{PAC}[T,S]$.
\end{definition}

\begin{definition} \label{def:classes}
Define the complexity classes as follows:
\begin{align*}
\fptpactime &= \mathsf{PAC}[\mathrm{fpt},\mathrm{poly}], \\
\fptpac &= \mathsf{PAC}[\mathrm{fpt},\mathrm{fpt}], \\
\xppactime &= \mathsf{PAC}[\mathrm{xp},\mathrm{poly}], \\
\xppac &= \mathsf{PAC}[\mathrm{xp},\mathrm{xp}],
\end{align*}
where we fix
\begin{align*}
\poly &= \poly(n,s,1/\varepsilon,1/\delta,k,\ell), \\
\fpt &= \fpt(n,s,1/\varepsilon,1/\delta;k,\ell), \\
\xp &= \xp(n,s,1/\varepsilon,1/\delta;k,\ell).
\end{align*}
\end{definition}

There are examples of natural problems falling into each of these classes~\cite{BrandGanianSimonov23}.
In addition to the above, there is a fifth class that may be considered here: $\mathsf{PAC}[\mathrm{xp},\mathrm{fpt}].$ However, we are not aware of any natural problems residing there that are not given by the ``lower'' classes.

Figure \ref{fig:classes} provides an overview of these complexity classes and their relationships.

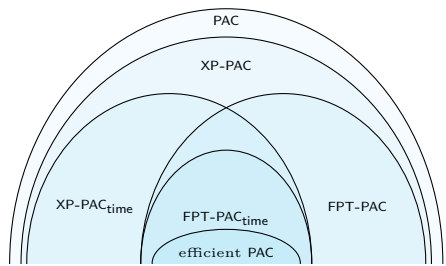
\begin{figure}
\begin{center}

\begin{tikzpicture}[set/.style={fill=cyan,fill opacity=0.04},scale=0.75]
\clip (-4.3,0) rectangle (5,5);

\draw[set] (0,0) ellipse (3.8cm and 4.5cm);
\node at (0,4.3) {\tiny $\mathsf{PAC}$};

\draw[set] (0,0) ellipse (3.6cm and 4cm);
\node at (0,3.5) {\tiny $\xppac$};

\node at (-2.3,1) {\tiny $\xppactime$};
\draw[set] (-1,0) ellipse (2.5cm and 3cm);

\node at (2.3,1) {\tiny $\fptpac$};
\draw[set] (1,0) ellipse (2.5cm and 3cm);

\node at (0,0.8) {\tiny $\fptpactime$};
\draw[set] (0,0) ellipse (1.5cm and 2cm);

\node at (0,0.2) {\tiny efficient $\mathsf{PAC}$}; 
\draw[set] (0,0) ellipse (1.3cm and 0.6cm);
\end{tikzpicture}
\end{center}
\caption{A schematic view of the parameterized learning classes defined in Definition \ref{def:classes}.}
\label{fig:classes} 
\end{figure}

\subsection{Consistency Checking for PAC-Learning}
We now recall the results tying the complexity of (parameterized) PAC-learning to (parameterized) consistency checking.
We have already shown that a consistency checking problem can be transformed into a learning problem, by viewing the hidden solution as the representation of the hidden concept; the same operation can also be done the other way around. Moreover, this transformation can be performed while respecting the parameters. Let $\mathcal{P}_{\text{cons}}$ be a consistency checking problem, and let $\mathcal{P}_{\text{learn}}$ be the respective learning problem. Consider a parameterized version $(\mathcal{P}_{\text{cons}}, \kappa + \lambda)$ of $\mathcal{P}_{\text{cons}}$, where $\kappa$ maps solutions $S \in \{0, 1\}^\ast$ to natural numbers, and $\lambda$ maps lists of labeled instances  $((I_1, \lambda_1), \ldots, (I_t, \lambda_t))$, $I_i \in \{0, 1\}^*$, $\lambda_i \in \{0, 1\}$, to natural numbers. The parameterized learning problem is then $(\mathcal{P}_{\text{learn}}, \kappa, \lambda')$, where $\kappa$ is given by the same function is the parameterization of representations, as representations of concepts are exactly the solutions in the original search problem, and $\lambda'(\mathcal{D})$ for a distribution $\mathcal{D}$ is the maximum value of $\lambda(\mathcal{L})$, where $\mathcal{L}$ is any set of labeled instances produced by sampling from~$\mathcal{D}$.

It is well-known that, under the assumption that the hypothesis space is not too large, there is an equivalence between a learning problem being PAC-learnable and the corresponding consistency checking problem being solvable in randomized polynomial time~\cite{PittV88}. Brand, Ganian and Simonov proved a generalization of this equivalence in the parameterized sense~\cite{BrandGanianSimonov23}, which we recall next

\begin{theorem}[Corollary of Theorem 3.17~\cite{BrandGanianSimonov23}] \label{thm:red-dec-to-learn}
    Let $\mathcal{P}_{\text{cons}}$ be a parameterized consistency checking problem, and $\mathcal{P}_{\text{learn}} = (\mathcal{C},\rho,\kappa,\lambda)$ be its matching parameterized learning problem,
where $\lambda$ depends only on the support of the distribution.

If $\mathcal{P}_{\text{learn}}$ is in $\fptpac$, then $\mathcal{P}_{\text{cons}}$ is in $\FPT$.

Similarly, if $\mathcal{P}_{\text{learn}}$ is in $\xppac$, then $\mathcal{P}_{\text{cons}}$ is in $\XP$.
\end{theorem}

\begin{theorem}[Corollary of Theorem 3.19~\cite{BrandGanianSimonov23}]\label{thm:red-learn-to-dec}
    Let $\mathcal{P}_{\text{cons}}$ be a parameterized consistency checking problem, and $\mathcal{P}_{\text{learn}} = (\mathcal{C},\rho,\kappa,\lambda)$ be its matching parameterized learning problem.
    Denote the set of representations of concepts in $C \in \mathcal{C}$ of arity $n$ with $\kappa(C) = k$ by $\mathcal{H}_{n,k}$.

If $\mathcal{P}_{\text{cons}}$ is in $\FPT$ and $\log |\mathcal{H}_{n,k}| \in \fpt(n;k)$,
then $\mathcal{L}$ is in $\fptpactime$.

Similarly, if $\mathcal{P}_{\text{cons}}$ is in $\XP$ and $\log |\mathcal{H}_{n,k}| \in \xp(n;k)$, then $\mathcal{L}$ is in $\xppactime$.
\end{theorem}

The theorems above allow us to automatically transfer parameterized algorithmic upper and lower bounds for consistency checking into upper and lower bounds for parameterized learning problems, respectively.
If a parameterized consistency checking  problem is efficiently solvable by a parameterized algorithm, by Theorem \ref{thm:red-learn-to-dec} we get that the parameterized learning problem is efficiently solvable. Note that in the problems considered in this paper the solution is always a set of vertices/edges, or a partition into such sets, thus $\log |\mathcal{H}_{n,k}|$ is always polynomial.

On the other hand, Theorem \ref{thm:red-learn-to-dec} tells us that an efficient algorithm for a parameterized learning problem implies an efficient algorithm for the corresponding paramerized consistency checking problem. Turning this around, we see that lower bounds on consistency checking imply lower bounds for learning. That is, if $\mathcal{P}_{\text{cons}}$ is \W[1]-hard, then $\mathcal{P}_{\text{learn}}$ is not in $\fptpactime$ unless $\FPT = \W[1]$.

\section{Partitioning Problems: 2-Coloring and Split Graphs}
\label{sec:partition}
We begin our investigation by using two basic vertex bipartition problems on graphs to showcase some of the unexpected complexity-theoretic behavior of consistency checking problems. Let us first consider $2$-\textsc{Coloring}, i.e., the problem of partitioning the vertex set into two independent sets. There are two natural ways one can formalize $2$-\textsc{Coloring} as a search problem: either one asks for a vertex set $X$ such that both $X$ and the set of vertices outside of $X$ are independent (i.e., they form a proper $2$-coloring), or one asks for two independent sets $X,Y$ which form a bipartition of the vertex set. Here, we consider the former variant since it has a slightly smaller hypothesis space\footnote{
In general, the precise definition of the sought-after object can be of great importance in the context of consistency checking; this is related to the well-known fact that the selection of a hypothesis space can have a fundamental impact on PAC learnability. However, in our case the proofs provided in this section can also be used to obtain the same results for the latter variant.}.

\problem{\CCcol}{$\mathcal{I}=\{V,(E_i, \lambda_i)_{i\in[t]}\}$ where for each $i\in [t]$, $G_i=(V,E_i)$ is a graph and $\lambda_i\in \{0,1\}$.}{
A set $X\subseteq V$ such that for each $i\in[t]$, $(X,V\setminus X)$ forms\\ a proper $2$-coloring of $G_i$ if and only if $\lambda_i=1$.}

As our first result, we show that \CCcol\ is \NP-hard.

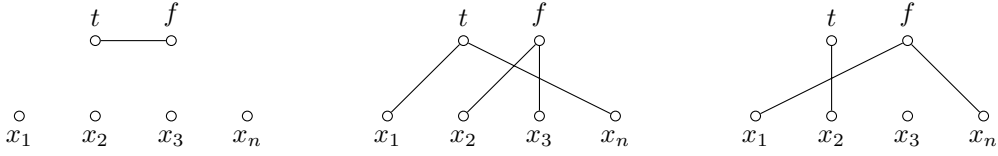
\begin{figure}
\begin{tikzpicture}

\node[circ, label=above:{$t$}] (t) at (1,1) {};
\node[circ, label=above:{$f$}] (f) at (2,1) {};
\node[circ, label=below:{$x_1$}] (1) at (0,0) {};
\node[circ, label=below:{$x_2$}] (2) at (1,0) {};
\node[circ, label=below:{$x_3$}] (3) at (2,0) {};
\node[circ, label=below:{$x_n$}] (n) at (3,0) {};

\draw (t) -- (f); 

\end{tikzpicture}
\hspace{1cm}
\begin{tikzpicture}

\node[circ, label=above:{$t$}] (t) at (1,1) {};
\node[circ, label=above:{$f$}] (f) at (2,1) {};
\node[circ, label=below:{$x_1$}] (1) at (0,0) {};
\node[circ, label=below:{$x_2$}] (2) at (1,0) {};
\node[circ, label=below:{$x_3$}] (3) at (2,0) {};
\node[circ, label=below:{$x_n$}] (n) at (3,0) {};

\draw (t) -- (1)
(f) -- (2)
(t) -- (n)
(f) -- (3); 

\end{tikzpicture}
\hspace{1cm}
\begin{tikzpicture}

\node[circ, label=above:{$t$}] (t) at (1,1) {};
\node[circ, label=above:{$f$}] (f) at (2,1) {};
\node[circ, label=below:{$x_1$}] (1) at (0,0) {};
\node[circ, label=below:{$x_2$}] (2) at (1,0) {};
\node[circ, label=below:{$x_3$}] (3) at (2,0) {};
\node[circ, label=below:{$x_n$}] (n) at (3,0) {};

\draw (t) -- (2)
(f) -- (n)
(f) -- (1); 

\end{tikzpicture}
\caption{For the \SAT\xspace instance $\varphi=(x_1\vee \overline{x_2}\vee \overline{x_3}\vee x_n)\wedge (\overline{x_1}\vee x_2\vee \overline{x_n})$, $n=4$ the correspondent \CCcol instance $\mathcal{I}=\{V, \{(E_+, 1), (E_{C_1}, 0), (E_{C_2}, 0)\}\}$, $V=\{t, f, x_1, x_2, x_3, x_n\}$.}
\label{fig:2col}
\end{figure}

\begin{theorem}
\label{thm:2colNP}
There is no polynomial-time algorithm that solves \CCcol
unless $\P=\NP$.
\end{theorem}

\begin{proof}
We present a reduction that takes an $n$-variable instance $\varphi$ of the
\SATpr problem (\SAT) and constructs an instance $\mathcal{I}$ of \CCcol\ which admits a solution if and only if $\varphi$ is satisfiable.
Let $\mathcal{C}$ denote the set of clauses of~$\varphi$.

\vspace{0.2cm}\noindent\textit{Construction.}\quad
First, we set the vertex set $V$ in $\mathcal{I}$ to be $\{f, t, x_1, x_2, \dots, x_n\}$.
For each clause $C\in \mathcal{C}$, we construct an edge set $E_C$ as follows.
For each $i\in [n]$, if a true (false) assignment of $x_i$ satisfies $C$, then we add the edge $tx_i$ ($fx_i$) to $E_C$.
For each such edge set $E_C$, we set $\lambda_C=0$.
Finally, we add to $\mathcal{I}$ a positive sample $(E_+, 1)$ such that $E_+=\{tf\}$. An illustration is provided in \Cref{fig:2col}.

\vspace{0.2cm}\noindent\textit{Correctness.}\quad
Suppose, given an instance $\varphi$ of \SAT, that the reduction
described above returns $\mathcal{I}=\{V, (E_i, \lambda_i)\}_{i\in \calC\cup \{+\}}$ as an instance of \CCcol.

Assume that $\varphi$ admits a satisfying assignment $\mathcal{A}\colon \{x_i\}_{i\in[n]}\rightarrow \{\true, \false\}$. 
Consider the colloring $\chi\colon V\rightarrow \{\blue, \red\}$ such that $\chi(t)=\red$, $\chi(f)=\blue$, and for each $i\in [n]$, $\chi(x_i)=\red$ if and only if $\mathcal{A}(x_i)=\true$.

First, the sample $(E_+, 1)$ of $\mathcal{I}$ is consistent with the coloring $\chi$, since its only edge $ft$ was colored properly.
Then, for each $C\in\mathcal{C}$, the sample $(E_C, 0)$ must be consistent with $\chi$, i.e., there exists at least one edge in $E_C$ with same colored endpoints.
Indeed, there must exist a variable $x_i$ is such that $\mathcal{A}(x_i)$ satisfies $\varphi$.

Then, by the construction of $\mathcal{I}$ instance, if $x_i=\true$ ($x_i=\false$) satisfies $C$ then $tx_i\in E_C$ ($fx_i\in E_C$) and hence both $x_i$ and $t$ are \red (both $x_i$ and $f$ are \blue) under the constructed coloring $\chi$.

For the other direction, suppose that there is a coloring $\chi\colon V\rightarrow \{\blue, \red\}$ that is consistent with the instance $\mathcal{I}$ of \CCcol. 
Then $\chi(t)\neq \chi(f)$ due to the construction of $(E_+, 1)\in\mathcal{I}$; 
without loss of generality, let $\chi(t)=\red$, $\chi(f)=\blue$. 
We retrieve a variable assignment $\mathcal{A}$ for $\varphi$ in the following way.
Recall that for each $C\in \mathcal{C}$, the coloring $\chi$ is consistent with the sample $(E_C, 0)$. 
Since the edge $ft$ has a proper coloring, at least one vertex $x_i$ has an edge to either $t$ or $f$ such that both its endpoints are colored the same way. 
If this edge is $x_if$ ($x_it$), then let $\mathcal{A}(x_i)=\false$ ($\mathcal{A}(x_i)=\true$). If this only results in a partial assignment, we extend this to a complete assignment of all variables in $\varphi$ by assigning the remaining variables arbitrarily.

We conclude by arguing that the resulting assignment $\mathcal{A}$ has to satisfy $\varphi$. 
Let us consider an arbitrary clause $C\in\mathcal{C}$ and an edge in the corresponding edge set $E_C$ with same colored endpoints, w.l.o.g.\ $x_if$.
Then, by the way we defined the assignment, $\mathcal{A}(x_i)=\false$.
But by our construction, the edge $x_if\in E_C$ only if $x_i=\false$ satisfies the clause $C$. Thus, the clause $C$ is satisfied by the assignment $\mathcal{A}$.
Following the same argument, each clause $C\in \mathcal{C}$, and accordingly the instance $\varphi$, is satisfied.

\end{proof}

It is worth noting that the graphs constructed by the reduction underlying \Cref{thm:2colNP} are very simple---in fact, even the graph induced by the union of all edges occurring in the instances of \CCcol\ produced by the reduction has a vertex cover number of $2$. This essentially rules out tractability via most standard structural graph parameters. A similar observation can also be made for most other consistency checking lower bounds obtained within this article.

As an immediate corollary of \Cref{thm:2colNP}, we obtain that the corresponding learning problem is not efficiently PAC-learnable~\cite{AlekhnovichBFKP08}. To provide a concrete example of the formal transition from consistency checking to the corresponding learning problem described in \Cref{sub:pac}, we state the problem: In \textsc{$2$-Coloring Learning}, we are given (1) a set $V$ of vertices, a confidence measure $\delta$ and an accuracy measure $\varepsilon$, (2) have access to an oracle that can be queried to return labeled samples of the form $(E,\lambda)$ where $E$ is an edge set over $V$ and $\lambda\in \{0,1\}$ according to some hidden distribution, and (3) are asked to return a vertex subset $X\subseteq V$, whereas a sample $E$ is evaluated as positive for $X$ if and only if $(X,V\setminus X)$ forms a $2$-coloring on $(V,E)$.

\begin{corollary}
\textup{\textsc{$2$-Coloring Learning}} is not efficiently PAC-learnable unless $\P=\NP$.\end{corollary}

While the intractability of consistency checking for \CCcol might already be viewed as surprising, let us now consider the related problem of partitioning the vertex set into one independent set and one clique---i.e., the \textsc{Split Graph} problem. As a graph search problem, \textsc{Split Graph} is well-known to be polynomially tractable~\cite{HammerS81}. Following the same line of reasoning as for \textsc{$2$-Coloring}, we formalize the corresponding search problem below. Let a pair of vertex subsets $(X\subseteq V,Y\subseteq V)$ be a \emph{split} in a graph $G=(V,E)$ if $(X,Y)$ is a bipartition of $V$ such that $X$ is a clique and $Y$ is an independent set.

\problem{\CCsplit}{$\mathcal{I}=\{V,(E_i, \lambda_i)_{i\in[t]}\}$ where for each $i\in [t]$, $G_i=(V,E_i)$ is a graph and $\lambda_i\in \{0,1\}$.}{
A set $X\subseteq V$ such that for each $i\in[t]$, $(X,V \setminus X)$ is a split in $G_i$\\ if and only if $\lambda_i=1$.}

Unlike \CCcol, \CCsplit turns out to be tractable.

\begin{theorem}
\label{thm:splitcons}
\CCsplit\ can be solved in time $\bigoh(|\mathcal{I}|^3)$.

\end{theorem}

\begin{proof}
Let us consider an input instance $\calI=\{V, (E_i, \lambda_i)_{i\in [t]}\}$ of \CCsplit.
The algorithm first checks whether $\calI$ contains at least one positive sample or consists of negative samples only; each of these two cases will be handled by a separate procedure and arguments.

If $\mathcal{I}$ contains at least one positive sample, say w.l.o.g.\ $(E_1, \lambda_1)$, then the algorithm enumerates all possible splits of $G_1=(V,E_1)$. Indeed, for each pair $(X_1, Y_1)$, $(X_2, Y_2)$ of splits it holds that $(X_2, Y_2)$ can be obtained from $(X_1, Y_1)$ by moving at most one vertex from the clique part to the independent part, and at most one vertex from the independent part to the clique part. Hence after computing an arbitrary split in linear time (e.g., from its degree sequence)~\cite{HammerS81}, the algorithm can enumerate the set of all splits of $G_1$ in at most quadratic time. The algorithm then checks, for each split of $G_1$, whether it is a solution for $\mathcal{I}$---in particular, whether it is a split for each positive sample and a non-split for each negative sample. This check can be done in linear time. For correctness of this case, it suffices to observe that every solution for $\mathcal{I}$ must necessarily be a split of $G_1$.

For the case where each sample in $\mathcal{I}$ is negative, we exploit the fact that the set of all splits in a graph can be enumerated in polynomial time in a different manner. In particular, for each sample $(E_j,0)$ in $\mathcal{I}$, we construct the set $Q_j$ of all splits of $G_j=(V,E_j)$ in at most quadratic time; in particular, $Q_j$ will be empty if $G_j$ is not a split graph, and otherwise will contain at most a quadratic number of splits. Let $Q=\bigcup_{j\in [t]}Q_j$, and observe that $|Q|\leq |V|^2\cdot t$. The algorithm then proceeds by enumerating, in brute force and in arbitrary order, all possible vertex $2$-partitions of $V$, and for each such $2$-partition $(X,V\setminus X)$ it checks whether $(X,V\setminus X)\in Q$ or not. If $(X,V\setminus X)\in Q$, then we proceed to the next $2$-partition, while otherwise we output $(X,V\setminus X)$ as the solution; if every $2$-partition turns out to be in $Q$ then the algorithm outputs that no solution exists. Clearly, this procedure terminates after at most $|Q|+1$ steps. For correctness, it suffices to observe that every $2$-partition is a solution for $\mathcal{I}$ if and only if it does not lie in $Q$---indeed, every $2$-partition in $Q$ is a split for at least one negative sample (and hence cannot be a solution), and every $2$-partition that is not in $Q$ is not a split for any negative sample and hence is a solution.

\end{proof}

Naturally, one can formalize the learning problem for \CCsplit\ in an analogous way as was done for \textsc{$2$-Coloring Learning}. Since the hypothesis bound of \Cref{thm:red-learn-to-dec} holds here as well, \Cref{thm:splitcons} implies:

\begin{corollary}
\textup{\textsc{Split Graph Learning}} is efficiently PAC-learnable.
\end{corollary}

Let us now conclude the section by revisiting the polynomial-time intractability of \CCcol\ through the lens of parameterized complexity theory. Naturally, there are many parameterizations one may consider in the setting---as an exercise that follows the same exhaustive-branching ideas as those used for \textsc{Vertex Cover}~\cite[Lemma 6.1]{BrandGanianSimonov23}, one could for instance attempt to parameterize by the size of the smaller color class in the sought-after coloring, whereas a fixed-parameter algorithm in this setting (based on exhaustive branching) would yield a \fptpactime\ algorithm for \textsc{$2$-Coloring Learning} in the corresponding parameterization of the concept. In this article, we instead showcase a less straightforward fixed-parameter algorithm for the problem when parameterized by the number of negative samples on the input (which in turn corresponds to a parameterization of the distribution in the learning setting~\cite{BrandGanianSimonov23}). It will later turn out that the same parameterization can be used to achieve fixed-parameter tractability for several other consistency checking problems as well, albeit the individual techniques used vary from problem to problem.

Let $t^-=|\{(E_i, \lambda_i)_{i\in[t]}~|~\lambda_i=0\}|$ be the number of negative samples in an input instance~$\mathcal{I}$.

\begin{theorem}
\label{thm:ColFPT}
\CCcol\ is fixed-parameter tractable when parameterized by the number $t^-$ of negative samples.
\end{theorem}

\begin{proof}
Let $G^+=(V,E^+)$ be the graph obtained from an input instance $\mathcal{I}=\{V,(E_i, \lambda_i)_{i\in[t]}\}$ of \CCcol\ by setting $E^+:=\bigcup_{i~|~\lambda_i=1}E_i$, i.e., $G^+$ is the (non-disjoint) union of all yes samples in $\mathcal{I}$. It will be useful to observe that the solution for $\mathcal{I}$ must also be a $2$-coloring for $G^+$. Let $\{C_1,\dots,C_\ell\}$ be the set of connected components of $G^+$; for each connected component $C_j$, $j\in \ell$, we check whether $C_j$ is bipartite. If any such component is not bipartite, we can correctly output that $\mathcal{I}$ does not admit a solution; otherwise we assume that each $C_j$ is bipartitioned into independent sets $(A_j,B_j)$ and proceed with the algorithm. 

Let us now consider a negative sample $(E_i,0)$. We define the \emph{signature} $S_i$ of $(E_i,0)$ as the set $\{\{X,Y\}~|~\exists xy\in E_i: (x\in X) \wedge (y\in Y) \wedge (X,Y\in \{A_1,\dots,A_\ell\}\cup \{B_1,\dots,B_\ell\}\}$; in other words, the signature is obtained by abstracting away the individual identies of the endpoints of $E_i$ and replacing this by information about which part of which component the endpoints belong to. We remark that the signature is treated as a set of two-element multisets to accommodate for the (admittedly trivial) case where $X=Y$.

We now distinguish two cases based on the size of the signature $S_i$ of each negative sample $(E_i,0)$. If $|S_i|>16(t^-)^2$ then we mark $S_i$ as \emph{large}, and otherwise we mark it as \emph{small}. Observe that each large signature must contain pairs from more than $2t^-$ components of $G^+$; indeed, the size of any signature that only contains pairs from $2t^-$ components of $G^+$ is upper-bounded by $(2t^-)^2\cdot 4$. We then perform exhaustive branching to identify a single pair $(X,Y)\in S_i$ in each small signature. In every branch, we proceed as as follows: 
\begin{enumerate}
\item For each pair $(X,Y)$ identified in a small signature, we add a new degree-$2$ vertex to  $G^+$ and make it adjacent to an arbitrary vertex in $X$ and an arbitrary vertex in $Y$. Observe that this ensures that a proper $2$-coloring of $G^+$ must use the same color for $X$ and $Y$. We mark all connected components which have been attached to a new degree-$2$ vertex in this way as \emph{used}.
\item We check if the subgraph of $G^+$ induced on the used connected components admits a proper $2$-coloring $\zeta$. If that is not the case, we proceed to the next branch. Otherwise, we apply the pigeon-hole principle to identify, for each large signature $S_j$, a pair $\{X_j,Y_j\}\in S_j$ such that either $X_j$ or $Y_j$ are from a component $C^*_j$ which is not marked as used and mark this component $C^*_j$ as used as well; this can be done via linear-time enumeration of all connected components since each large signature must contain pairs from more than $2t^-$ components of $G^+$. 
\item We now expand the proper $2$-coloring $\zeta$ to a proper $2$-coloring $\upsilon$ of all of $G^+$ as follows. We drop the coloring of the auxiliary degree-$2$ vertices constructed in the first step (as these are not part of $G^+$). Each connected component of $G^+$ will be properly $2$-colored, but for each pair $\{X_j,Y_j\}$ in a large signature $S_j$ selected in the previous step we ensure that $\upsilon$ uses the same color for $X_j$ and $Y_j$. this fully determines the two color classes for all connected components of $G^+$. We check that $\upsilon$ is a solution for $\mathcal{I}$ and output it if that is the case.
\end{enumerate}

This concludes the description of the algorithm. The running time is upper-bounded by $\bigoh(2^{16(t^-)^2}\cdot |\mathcal{I}|)$, since the exhaustive branching preceding the three steps listed above loops over $2^{16(t^-)^2}$ cases and the three steps can be carried out in linear time. For correctness, observe that every solution provided by the algorithm is a solution for $\mathcal{I}$. On the other hand, assume that $\mathcal{I}$ admits a solution. For each negative sample $(E_i,0)$ with a small signature, there must be at least one edge $x_iy_i\in E_i$ such that both $x_i$ and $y_i$ belong to the same part in the solution, and let $\{X_i,Y_i\}$ be the corresponding tuple in $S_i$. Consider the branch of the algorithm which selects $\{X_i,Y_i\}\in S_i$. The supergraph of $G^+$ constructed in this branch now admits a $2$-coloring $\zeta$. At that point the algorithm is guaranteed to succeed in finding a solution for $\mathcal{I}$, as the constructed $2$-coloring $\upsilon$ will be proper for all positive samples (as it is proper for $G^+$), will not be proper for negative samples with a small signature (as ensured already by $\zeta$), and will not be proper for negative samples with a large signature either (as guaranteed in the final step when constructing~$\upsilon$).
\end{proof}

\section{Consistency Checking for Selected Edge Search Problems}
\label{sec:edgesearch}

In this section, we perform a parameterized analysis of consistency checking for three natural and extensively studied edge search problems on graphs: \textsc{Matching}, \textsc{($k$-)Path} and \textsc{Edge Clique Cover}. We formalize the parameterized consistency checking formulations of these three problems below; recall that a set $\mathcal{F}=\{F_1,\dots,F_\ell\}$ is an \emph{edge clique cover} if each $F_i$, $i\in [\ell]$ is the edge set of a clique in the graph and each edge in the graph is contained in at least one $F_i$, $i\in [\ell]$~\cite[Subsection 2.2.3]{DBLP:books/sp/CyganFKLMPPS15}. 

\problempar{\CCmatch}{$\mathcal{I}=\{V,(E_i, \lambda_i)_{i\in[t]}\}$ where for each $i\in [t]$, $G_i=(V,E_i)$ is a graph and $\lambda_i\in \{0,1\}$,\\ and an integer $k$.}{
$k$.}{A set $F\subseteq {V\choose 2}$ of size $k$ such that for each $i\in[t]$, $F$ forms\\ a matching in $G_i$ if and only if $\lambda_i=1$.}

\problempar{\CCpath}{$\mathcal{I}=\{V,(E_i, \lambda_i)_{i\in[t]}\}$ where for each $i\in [t]$, $G_i=(V,E_i)$ is a graph and $\lambda_i\in \{0,1\}$,\\ and an integer $k$.}{
$k$.}{A set $F\subseteq {V\choose 2}$ of size $k$ such that for each $i\in[t]$, $F$ forms\\ a path in $G_i$ if and only if $\lambda_i=1$.}

\problempar{\CCEQC}{$\mathcal{I}=\{V,(E_i, \lambda_i)_{i\in[t]}\}$ where for each $i\in [t]$, $G_i=(V,E_i)$ is a graph and $\lambda_i\in \{0,1\}$,\\ and an integer $k$.}{
$k$.}{A set $\mathcal{F}\subseteq 2^{V\choose 2}$ of size $k$ such that for each $i\in[t]$, $\mathcal{F}$ forms\\ an edge clique cover in $G_i$ if and only if $\lambda_i=1$.}

An observant reader may notice that in the first of the three problems above, we consider solution size as a parameter even though the corresponding search problem of finding a maximum matching in a graph is polynomial-time tractable. This is due to the fact that, as it turns out, \textsc{Matching} in the consistency checking regime is not polynomial-time tractable unless $\P=\NP$. In fact, we show an even stronger (and more surprising) result:

\begin{theorem}
\label{thm:CCmatchhard}
\CCmatch does not admit a fixed-parameter algorithm unless $\FPT=\W[2]$.
\end{theorem}

\begin{proof}

We present a reduction that given an instance $(\calU, \calF, k')$ of the classical \SCpr problem~\cite{DBLP:books/sp/CyganFKLMPPS15}, constructs an instance $\mathcal{I}$ of \CCmatch\ which admits a solution if and only if $(\calU, \calF, k')$ is a yes-instance.
An instance $(\calU, \calF, k')$ of \SCpr is a family $\calF=\{F_1,\dots,F_m\}$ of $m$ subsets over the $n$-element universe $\calU=\{u_1,\dots,u_n\}$, and we are asked whether there exists a $k'$-element subset of $\calF$ whose union contains all of $\calU$.

\vspace{0.2cm}\noindent\textit{Construction.}\quad
We construct the instance $\mathcal{I}=\{V,(E_i, \lambda_i)_{i\in[t]}\}$ of \CCmatch\ as follows, with the parameter set to $k=k'$. Let the unique positive sample in $\mathcal{I}$ be the edge set $E_1$ such that the graph $(V,E_1)$ is a set of $k$ disjoint stars, whereas for each $i\in [k]$ the graph $(V,E_1)$ contains a center $s^i$ adjacent to pendants $p^i_1,\dots,p^i_m$. Next, for each element $u_j\in \calU$, $j\in [n]$, we add a negative sample $(V,E_{j+1})$ into $\mathcal{I}$ which only contains non-edges between the centers of stars and the leaves (of the same star) corresponding to the sets containing that element; formally, $E_{j+1}={V\choose 2}\setminus \{s^ip^i_{\ell}~|~i\in[k], u_j\in F_\ell\}$. This completes the construction of $\mathcal{I}$ (see also \Cref{fig:matching}).

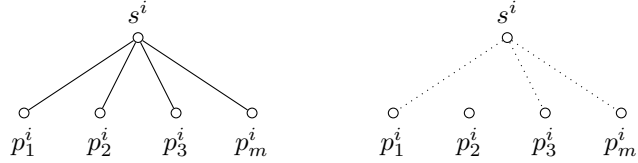
\begin{figure}[ht]
\centering
\begin{tikzpicture}

\node[circ, label=above:{$s^i$}] (t) at (1.5,1) {};
\node[circ, label=below:{$p^i_1$}] (1) at (0,0) {};
\node[circ, label=below:{$p^i_2$}] (2) at (1,0) {};
\node[circ, label=below:{$p^i_3$}] (3) at (2,0) {};
\node[circ, label=below:{$p^i_m$}] (m) at (3,0) {};

\draw (t) -- (1)
(t) -- (2)
(t) -- (3)
(t) -- (m); 

\end{tikzpicture}
\hspace{1cm}
\begin{tikzpicture}

\node[circ, label=above:{$s^i$}] (t) at (1.5,1) {};
\node[circ, label=below:{$p^i_1$}] (1) at (0,0) {};
\node[circ, label=below:{$p^i_2$}] (2) at (1,0) {};
\node[circ, label=below:{$p^i_3$}] (3) at (2,0) {};
\node[circ, label=below:{$p^i_m$}] (m) at (3,0) {};

\draw[dotted] (t) -- (1)
(t) -- (3)
(t) -- (m);
\end{tikzpicture}
\caption{Reducing from \SCpr, the \CCmatch instance has a positive sample with $k$ ($i\in [k]$) stars as shown on the left; for each $u_j\in \calU$, the correspondent \no-instance is a complete graph but excluding $\{s^ip^i_{\ell}~|~i\in[k], u_j\in F_\ell\}$. So, as an example, if $m=4$ and $u_{j}\in F_{\ell}$ for any $\ell\in \{1, 3, m\}$, then  for all $i\in [k]$, the dotted edges are out of the construction.}
\label{fig:matching}
\end{figure}

\vspace{0.2cm}\noindent\textit{Correctness.}\quad
If $\mathcal{I}$ admits a solution $Q$, then $Q$ must be a matching in $(V,E_1)$ of size $k$ and hence can only contain a single edge from each of the $k$ stars. Hence, $Q=\{s^1p_{\alpha(1)}, s^2p_{\alpha(2)},\dots, s^kp_{\alpha(k)}\}$ for some mapping $\alpha$. Moreover, since $Q$ is not a matching in $(V,E_{j+1})$ for any $j\in [n]$, the set $\{F_{\alpha(1)},\dots,F_{\alpha(k)}\}$ is a set cover for $(\calU, \calF, k)$. At the same time, given a set cover $\{F_{\beta(1)},\dots,F_{\beta(k)}\}$ (for some mapping $\beta$), we can construct a solution for $\mathcal{I}$ by taking $\{s^1p_{\beta(1)},\dots,s^kp_{\beta(k)}\}$. This yields a reduction from \SCpr\ to the problem of deciding the existence of a solution for \CCmatch; in particular, this means that a fixed-parameter algorithm for \CCmatch\ would imply \FPT=W[2].
\end{proof}

A similar reduction also allows us to establish the intractability of consistency checking for \textsc{Path}.

\begin{theorem}
\label{thm:CCpathhard}
\CCpath does not admit a fixed-parameter algorithm unless $\FPT=\W[2]$.
\end{theorem}

\begin{proof}
We present a reduction that given an instance $(\calU, \calF, k')$ of \SCpr, constructs an instance $\mathcal{I}$ of \CCpath\ which admits a solution if and only if $(\calU, \calF, k)$ is a yes-instance.
Recall that an instance $(\calU, \calF, k)$ of \SCpr is a family $\calF=\{F_1,\dots,F_m\}$ of $m$ subsets over the $n$-element universe $\calU=\{u_1,\dots,u_n\}$, and we are asked whether there exists a $k$-element subset of $\calF$ whose union contains all of $\calU$.

\vspace{0.2cm}\noindent\textit{Construction.}\quad
We construct the instance $\mathcal{I}=\{V,(E_i, \lambda_i)_{i\in[t]}$ of \CCpath\ as follows, with the parameter $k$ set to $2k'$. Let the unique positive sample in $\mathcal{I}$ be the edge set $E_1$ such that the graph $(V,E_1)$ is a $(2k+1)$-partite graph consisting of independent sets $S_1,\dots,S_{2k+1}$. For each odd $j\in [2k+1]$, the set $S_j$ contains a single vertex $s_j$. For each even $j\in[2k+1]$, the set $S_j$ contains one vertex for each set in $\calF$, and in particular $S_j:=\{p^j_1,\dots,p^j_m\}$. $E_1$ is then the set of all edges connecting consecutive sets in this partition of $V$, i.e., $E_1=\{s_{2j-1}p^{2j}_\ell~|~\ell\in[m], j\in [k]\} \cup \{p^{2j}_\ell s_{2j+1}~|~\ell\in[m], j\in[k]\}$ (see in \Cref{fig:path}). Next, for each element $u_j\in \calU$, $j\in [n]$, we add a negative sample $(V,E_{j+1})$ into $\mathcal{I}$ which only contains non-edges incident to the sets containing that element; formally, $E_{i+1}={V\choose 2}\setminus \{vp^{2j}_{\ell}~|~j\in[k], u_i\in F_\ell\}$. This completes the construction of $\mathcal{I}$.

\begin{figure}[ht]
\centering
\begin{tikzpicture}

\node[circ, label=left:{$s_1$}] (1) at (0,2) {};
\node[circ, label=above:{$p^2_1$}] (p1) at (1,4) {};
\node[circ, label=above:{$p^2_2$}] (p2) at (1,3) {};
\node[label=above:{$\vdots$}] (p3) at (1,1.15) {};
\node[circ, label=below:{$p^2_m$}] (p4) at (1,1) {};

\node[circ, label=below:{$s_3$}] (2) at (2,2) {};
\node[circ, label=above:{$p^4_1$}] (c1) at (3,4) {};
\node[circ, label=above:{$p^4_2$}] (c2) at (3,3) {};
\node[label=above:{$\vdots$}] (c3) at (3,1.15) {};
\node[circ, label=below:{$p^4_m$}] (c4) at (3,1) {};

\node[circ, label=below:{$s_5$}] (3) at (4,2) {};

\node (q1) at (4.5,3.5) {};
\node (q2) at (4.5,2.5) {};
\node (q4) at (4.5,1.5) {};
\node (w1) at (5.5,3.5) {};
\node (w2) at (5.5,2.5) {};
\node (w4) at (5.5,1.5) {};

\node[label=above:{$\dots$}] (t) at (5,2) {};

\node[circ, label=below:{$s_{2k-1}$}] (4) at (6,2) {};
\node[circ, label=above:{$p^{2k}_1$}] (r1) at (7,4) {};
\node[circ, label=below:{$p^{2k}_2$}] (r2) at (7,3) {};
\node[label=above:{$\vdots$}] (r3) at (7,1.15) {};
\node[circ, label=below:{$p^{2k}_m$}] (r4) at (7,1) {};
\node[circ, label=right:{$s_{2k+1}$}] (5) at (8,2) {};

\draw (1) -- (p1)
(1) -- (p2)
(1) -- (p4)
(2) -- (p1)
(2) -- (p2)
(2) -- (p4)
(2) -- (c1)
(2) -- (c2)
(2) -- (c4)
(3) -- (c1)
(3) -- (c2)
(3) -- (c4)
(3) -- (q1)
(3) -- (q2)
(3) -- (q4)
(4) -- (w1)
(4) -- (w2)
(4) -- (w4)
(4) -- (r1)
(4) -- (r2)
(4) -- (r4)
(5) -- (r1)
(5) -- (r2)
(5) -- (r4);
\end{tikzpicture}
\caption{A unique positive sample for an instance of \CCpath.}
\label{fig:path}
\end{figure}
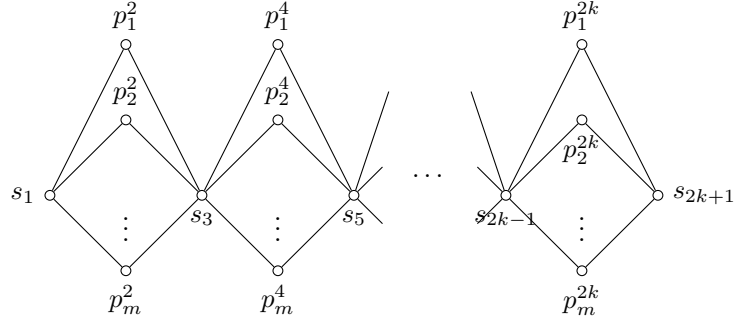

\vspace{0.2cm}\noindent\textit{Correctness.}\quad
If $\mathcal{I}$ admits a solution $Q$, then $Q$ must be a path in $(V,E_1)$ of size $2k$ and hence must be a path which consecutively visits vertices $(s_1,p^2_{\alpha(1)},s_3,p^4_{\alpha(2)},\dots,p^{2k}_{\alpha(k)},s_{2k+1})$. Moreover, since $Q$ is not a path in $(V,E_{j+1})$ for any $j\in [n]$, the set $\{F_{\alpha(1)},\dots,F_{\alpha(k)}\}$ is a set cover for $(\calU, \calF, k)$. At the same time, given a set cover $\{F_{\beta(1)},\dots,F_{\beta(k)}\}$ (for some mapping $\beta$), we can construct a solution for $\mathcal{I}$ by taking $(s_1,p^2_{\beta(1)},s_3,p^4_{\beta(2)},\dots,p^{2k}_{\beta(k)},s_{2k+1})$. This yields a reduction from \SCpr\ to the problem of deciding the existence of a solution for \CCpath; in particular, this means that a fixed-parameter algorithm for \CCpath\ would imply \FPT=W[2].
\end{proof}

However, we show that the third problem under consideration---\textsc{Edge Clique Cover}---does not become more difficult in the consistency checking regime.

\begin{theorem}
\label{thm:CCEDS}
\CCEQC admits a fixed-parameter algorithm which runs in time $\bigoh(2^{2^k}\cdot |\mathcal{I}|)$.
\end{theorem}

\begin{proof}
We begin by observing that an edge clique cover $\mathcal{F}\subseteq 2^{V\choose 2}$ for a graph $(V,E_1)$ cannot be an edge clique cover for any graph $(V,E_2)$ such that $E_2\neq E_1$. Indeed, for each edge $e\in E_2\setminus E_1$ it holds that $e$ cannot be covered by $\mathcal{F}$, while every edge $e\in E_1\setminus E_2$ must be covered by some clique $F\in \mathcal{F}$ in $(V,E_1)$ and hence $F$ would no longer be a clique in $(V,E_2)$. Hence if an instance $\mathcal{I}$ of \CCEQC contains two distinct positive samples, the algorithm can correctly output that $\mathcal{I}$ has no solution.

So, let us consider the case where $\mathcal{I}$ contains precisely one positive sample, say $(E_1, 1)$.
There is a well-known set of simple rules that can reduce the number of vertices of $(V,E_1)$ to $2^k$ \cite{Gramm_ECC}; in fact, under the Exponential Time Hypothesis~\cite{ImpagliazzoPZ01} this reduction combined with a brute-force exhaustive search on the reduced instance produces an essentially optimal algorithm for \textsc{Edge Clique Cover}~\cite{CPP_ECC}. We apply this procedure to either identify, in $2^{2^k}\cdot |V|$ time, an edge clique cover $\mathcal{F}$ of cardinality $k$ for $(E_1, 1)$, or correctly determine that no such $\mathcal{F}$ exists (in which case the algorithm can, as before, correctly output that $\mathcal{I}$ has no solution). The algorithm then simply checks to ensure $\mathcal{F}$ is not an edge clique cover for any of the negative samples, or equivalently, checks that there is no negative sample of the form $(E_1,0)$. If this check succeeds, the algorithm outputs the solution $\mathcal{F}$.

Finally, for the case where $\mathcal{I}$ contains only negative samples, let us consider $\calG=(V, \bigcap_{i\in[t]}E_i)$. If $|V|\leq k$, we can apply exhaustive branching to check each potential choice of $\mathcal{F}$, resulting in an algorithm with running time $2^{2^{k}}$. Otherwise, choose an arbitrary set of distinct vertices $v_1,\dots,v_k\in V$. Set $\mathcal{F}=\{F_1,\dots,F_k\}$ where $F_1=\{v_1v_2,v_1v_3\}$, and for each $F_i$, $2\leq i\leq k$ we use $F_i=\{v_iv_{i+1}\}$.
\end{proof}

Given the fixed-parameter intractability of \CCmatch and \CCpath\ w.r.t.\ the solution size alone, it is natural to ask whether one could solve these problems at least when the number of negative samples is small, similarly as was done in~\Cref{thm:ColFPT} for $2$-\textsc{Coloring}. We conclude this section by answering this question positively, albeit the algorithmic techniques used here are different from~\Cref{thm:ColFPT}. In fact, it turns out that an adaptation of the classical color-coding technique suffices in this case~\cite[Subsections 5.2 and 5.6]{DBLP:books/sp/CyganFKLMPPS15}.
For both problems, the task essentially boils down to intersecting all positive samples into one, and then looking for a solution where the set of $k$ edges is not contained in any negative sample. After assuming that all vertices of the solution receive distinct colors, we can perform dynamic programming to find a colorful solution, and while doing so we also store information about which negative samples are already ``dealt with'', i.e., which negative samples do not contain the edges in the partial solution.
We provide the formal proofs of both results below.

\begin{theorem}
\label{thm:CCmatchfpt}
\CCmatch admits an algorithm which runs in time $2^{\bigoh(k + t^-)} \cdot |\mathcal{I}|^{\bigoh(1)}$; in particular, the problem is fixed-parameter tractable when parameterized by $k+t^-$.
\end{theorem}

\begin{proof}
    Let $(V, k, (E_1, \lambda_1), \ldots, (E_t, \lambda_t))$, be the input of \CCmatch; w.l.o.g. assume that $\lambda_1 = \ldots = \lambda_{t^-} = 0$, and $\lambda_{t^- + 1} = \ldots = \lambda_{t} = 1$.
    First, consider the special case where $t^- = t$. If $|V| \ge 3$, return a set $F$ of two vertex pairs that share a vertex, which ensures that $F$ is not a matching in any sample. If $|V| \le 2$, the problem is trivial.

    We now have that $t^- < t$. Clearly, a set $F \subseteq \binom{V}{2}$ of size $k$ is a solution if and only if (1) the pairs in $F$ are vertex-disjoint, (2) for every $i \in [t^- + 1, t]$, $F \subseteq E_i$, and (3) for every  $i \in [t^-]$, $F \nsubseteq E_i$.
    The instance is thus equivalent to $(V, k, (E_1, 0), \ldots, (E_{t^-}, 0), (E, 1))$, where $E = E_{t^- + 1} \cap \cdots \cap E_t$, i.e., an instance obtained from the original one by intersecting all positive samples. A solution is now a matching $F$ of size $k$ in $(V, E)$ such that for every  $i \in [t^-]$, $F \nsubseteq E_i$.

    Assume now that the vertex set $V$ is colored in $2k$ colors with a mapping $c: V \to [2k]$, and we are looking for a colorful solution, i.e., a matching $M$ such that in the set of vertices of $M$ each color appears exactly once, in addition to the properties of the solution above. We show how to find such a colorful solution with the help of dynamic programming.

    For $C \subseteq [2k]$, $I \subseteq[t^-]$ let $\alpha(C, I)$ be $1$ if there exists a matching $M$ in $(V, E)$ such that its vertex set has exactly the colors in $C$, and such that $i \in I$ if and only if $M \nsubseteq E_i$; $\alpha(C, I) = 0$ if there is no such matching. We compute the values $\alpha(C, I)$ in the order of increasing $|C|$. To initialize, we set $\alpha(\emptyset, \emptyset) = 1$ and $\alpha(\emptyset, I) = 0$ for all $I \subseteq [t^-]$, $I \ne \emptyset$. This is clearly correct, as the only matching whose vertex set contains no colors is the empty matching, and the empty matching is a subset of every set of edges $E_i$, $i \in [t^-]$.

    Now consider $C \subseteq[2k]$ with $|C| > 0$, $I \subseteq [t^-]$, and assume $\alpha(C', I')$ is correctly computed for all $|C'| < |C|$ and all $I' \subseteq [t^-]$. For an edge $e \in E$, denote by $\beta(e) \subseteq [t^-]$ the set of negative samples that are ``dealt with'' by $e$, i.e., $\beta(e) = \{i \in [t^-] : e \notin E_i\}$. We set
    \begin{equation}
    \alpha(C, I) = \max_{\substack{uv \in E : c(u), c(v) \in C \\ I' \subseteq I: I' \cup \beta(uv) = I}} \alpha(C \setminus\{c(u), c(v)\}, I').
        \label{eq:dp_matching}
    \end{equation}

    We now argue the correctness of the computation above. First, assume $\alpha(C, I) = 1$, and consider a corresponding solution $M$. Let $uv \in M$ be an arbitrary edge of the solution, and let $I'$ be the set of negative samples ``dealt with'' by the remaining edges of $M$, i.e., $I' = \bigcup_{e \in M \setminus \{uv\}} \beta(e)$.
    Since $M$ is a colorful matching, the endpoints of edges in $M \setminus \{uv\}$ have colors in $C \setminus \{c(u), c(v)\}$. Therefore, $\alpha(C \setminus\{c(u), c(v)\}, I') = 1$ as $M \setminus \{uv\}$ is a suitable solution. Then also $\alpha(C, I) = 1$  by \eqref{eq:dp_matching}, since $uv \in E$, $c(u), c(v) \in C$ and $I' \cup \beta(uv) = I$, which implies that $\alpha(C \setminus\{c(u), c(v)\}, I')$ appears on the right-hand side of~\eqref{eq:dp_matching}.

    In the other direction, let $\alpha(C, I) = 0$, and assume for the sake of contradiction that \eqref{eq:dp_matching} assigns $1$ to it, i.e., there exists $uv \in E$ and $I' \subseteq I$ such that $c(u), c(v) \in C$, $I' \cup \beta(uv) = I$, and $\alpha(C \setminus\{c(u), c(v)\}, I') = 1$. Consider then the matching $M'$ certifying $\alpha(C \setminus\{c(u), c(v)\}, I') = 1$, and let $M$ be $M' \cup \{uv\}$. Clearly, $M$ is colorful and is a matching, as the edges in $M'$ have their endpoints' colors in $C \setminus \{c(u), c(v)\}$. Moreover, the set of colors covered by the edges of $M$ is then exactly $C$, and the set of ``dealt with'' negative samples is exactly $I$, as $I = I' \cup \beta(uv)$, where the former are satisfied by $M'$ and the latter by the edge $uv$. Therefore, we reach a contradiction that $\alpha(C, I)$ is $0$, which finishes the proof of correctness.

    Finally, a standard color-coding argument shows that solving the colorful version of the problem is sufficient, and this step could also be done in deterministic fashion with the claimed running time bound~\cite[Subsections 5.2 and 5.6]{DBLP:books/sp/CyganFKLMPPS15}. Observe that the dynamic programming above contains at most $2^k \cdot 2^{t^-}$ states $\alpha(C, I)$, and each is computed in time at most $2^{t^-} \cdot n^{O(1)}$ by~\eqref{eq:dp_matching}.
    
\end{proof}

The proof of the next theorem builds on the color-coding algorithm for $k$\textsc{-Path}, but otherwise the arguments are fairly similar to those used in the previous theorem.

\begin{theorem}
\label{thm:CCpathfpt}
\CCpath admits an algorithm which runs in time $2^{\bigoh(k + t^-)} \cdot |\mathcal{I}|^{\bigoh(1)}$; in particular, the problem is fixed-parameter tractable when parameterized by $k+t^-$.
\end{theorem}

\begin{proof}
    Let $(V, k, (E_1, \lambda_1), \ldots, (E_t, \lambda_t))$, be the input of \CCpath; w.l.o.g. assume that $\lambda_1 = \ldots = \lambda_{t^-} = 0$, and $\lambda_{t^- + 1} = \ldots = \lambda_{t} = 1$.
    First, assume $t^- = t$. If $|V| \ge 4$, return a set $F$ of two vertex pairs that do not share a vertex; clearly such $F$ is not a path in any instance and hence is a solution. On the other hand, the case where $|V| \le 3$ is trivial.

    We now have that $t^- < t$.
    By an analogous arguments as above, we can intersect all positive samples and arrive at the following equivalent formulation: A solution is a path $P$ of length $k$ in $(V, E)$ such that for every  $i \in [t^-]$, $P \nsubseteq E_i$, i.e., there is an edge in $P$ that is not present in $E_i$. Here, $E = E_{t^- + 1} \cap \cdots \cap E_t$.

    We again look for a colorful solution: assume that the vertex set $V$ is colored in $k + 1$ colors with a mapping $c: V \to [k + 1]$, and the task is to find a path $P$ such that in the set of vertices of $P$ each color appears exactly once, in addition to the properties of the solution above. Next we present a dynamic programming algorithm that finds such a colorful solution.

    For a vertex $v \in V$, $C \subseteq [k + 1]$, $I \subseteq[t^-]$ let $\alpha(v, C, I)$ be $1$ if there exists a path $P$ in $(V, E)$ such that its endpoint is $v$, its vertex set has exactly the colors in $C$, and such that $i \in I$ if and only if $P \nsubseteq E_i$; $\alpha(C, I) = 0$ if there is no such path. We compute the values $\alpha(v, C, I)$ in the order of increasing $|C|$. To initialize, for every vertex $v \in V$ we set $\alpha(v, \{c(v)\}, \emptyset) = 1$ and $\alpha(v, \{c(v)\}, I) = 0$ for all $I \subseteq [t^-]$, $I \ne \emptyset$. This is correct since every path with a single vertex is characterized by this vertex; each such path has an empty set of edges, and it is a subset of every set of edges $E_i$, $i \in [t^-]$.

    Now consider $v \in V$, $C \subseteq[k + 1]$ with $|C| > 0$, $I \subseteq [t^-]$, and assume $\alpha(v', C', I')$ is correctly computed for all $|C'| < |C|$ and all $v' \in V$, $I' \subseteq [t^-]$. For an edge $e \in E$, denote by $\beta(e)$ the set of negative samples that are ``dealt with'' by $e$, i.e., $\beta(e) = \{i \in [t^-] : e \notin E_i\}$. We set
    \begin{equation}
    \alpha(v, C, I) = \max_{\substack{uv \in E : c(u) \in C \\ I' \subseteq I: I' \cup \beta(uv) = I}} \alpha(u, C \setminus\{c(v)\}, I'),
        \label{eq:dp_path}
    \end{equation}
    if $c(v) \in C$; otherwise $\alpha(v, C, I)$ is clearly $0$.

    Next we show the correctness of the computation above. First, assume $\alpha(v, C, I) = 1$, and consider a corresponding solution $P$. Let $uv \in P$ be the final edge of the path, and let $I'$ be the set of negative samples ``dealt with'' by the remaining edges of $P$, i.e., $I' = \bigcup_{e \in P \setminus \{uv\}} \beta(e)$.
    Since $P$ is a colorful path, the endpoints of edges in $P \setminus \{uv\}$ have colors in $C \setminus \{c(v)\}$. Therefore, $\alpha(u, C \setminus\{c(v)\}, I') = 1$ as $P \setminus \{uv\}$ is a suitable solution. Then also $\alpha(v, C, I) = 1$  by \eqref{eq:dp_path}, since $uv \in E$, $c(u) \in C$ and $I' \cup \beta(uv) = I$, which implies that $\alpha(u, C \setminus\{c(v)\}, I')$ appears on the right-hand side of~\eqref{eq:dp_path}.

    In the other direction, let $\alpha(v, C, I) = 0$, and assume for the sake of contradiction that \eqref{eq:dp_path} assigns $1$ to it, i.e., there exists $uv \in E$ and $I' \subseteq I$ such that $c(u) \in C$, $I' \cup \beta(uv) = I$, and $\alpha(u, C \setminus\{c(v)\}, I') = 1$. Consider then the path $P'$ certifying $\alpha(u, C \setminus\{c(v)\}, I') = 1$, and let $P$ be $P' \cup \{uv\}$. Clearly, $P$ is colorful and is a path, as the edges in $P'$ have their endpoints' colors in $C \setminus \{c(v)\}$. Moreover, the set of colors covered by the edges of $P$ is exactly $C$, and the set of satisfied negative samples is exactly $I$, as $I = I' \cup \beta(uv)$. Therefore, we reach a contradiction that $\alpha(v, C, I)$ is $0$, which finishes the proof of correctness.

    A standard color-coding argument also shows that solving the colorful version of the problem is sufficient, and this step could also be done in deterministic fashion in the claimed running time bound~\cite[Subsections 5.2 and 5.6]{DBLP:books/sp/CyganFKLMPPS15}.
\end{proof}

\section{Consistency Checking for Selected Vertex Search Problems}
In the final technical section of this article, we focus our attention on consistency checking for two prominent vertex search problems in parameterized algorithmics: \textsc{Independent Set} and \textsc{Dominating Set}. As mentioned in the introduction, both problems are believed to be fixed-parameter intractable (the former is \W[1]-hard while the latter is \W[2]-hard), and so for the purposes of this article we restrict our attention to bounded-degree input graphs---or, more precisely, we consider the maximum degree as an additional parameter\footnote{We remark that all of the obtained results and proofs carry over also to the case where the maximum degree is considered to be an arbitrary fixed constant}. We formalize the consistency checking problems below.

\problempar{\CCIS}{Integers $k,d$, and $\mathcal{I}=\{V,(E_i, \lambda_i)_{i\in[t]}\}$ where for each $i\in [t]$, $G_i=(V,E_i)$ is a graph\\ of degree at most $d$ and $\lambda_i\in \{0,1\}$.}{
$k+d$.}{A set $X\subseteq V$ of size $k$ such that for each $i\in[t]$, $X$ forms\\ an independent set in $G_i$ if and only if $\lambda_i=1$.}

\problempar{\CCDS}{Integers $k,d$, and $\mathcal{I}=\{V,(E_i, \lambda_i)_{i\in[t]}\}$ where for each $i\in [t]$, $G_i=(V,E_i)$ is a graph\\ of degree at most $d$ and $\lambda_i\in \{0,1\}$.}{
$k+d$.}{A set $X\subseteq V$ of size $k$ such that for each $i\in[t]$, $X$ forms\\ a dominating set in $G_i$ if and only if $\lambda_i=1$.}

Once again, the complexity-theoretic properties of these problems turn out to be very different from those of their simpler graph search analogues. In particular, consistency checking for \textsc{Independent Set} is fundamentally harder than for the other two problems.

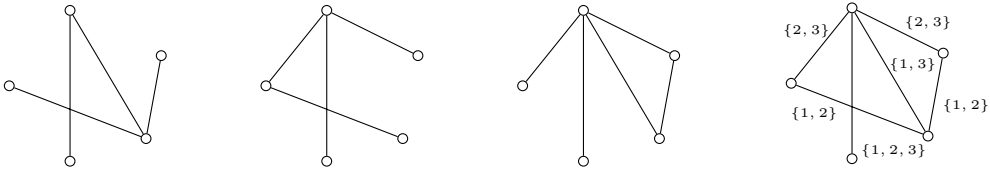
\begin{figure}
\label{fig:is}
\begin{tikzpicture}

\node[circ] (1) at (0,1) {};
\node[circ] (2) at (.8,2) {};
\node[circ] (3) at (2,1.4) {};
\node[circ] (4) at (1.8,.3) {};
\node[circ] (5) at (.8,0) {};

\draw (1) -- (4)
(4) -- (3)
(2) -- (4)
(2) -- (5)
; 

\end{tikzpicture}
\hspace{1cm}
\begin{tikzpicture}

\node[circ] (1) at (0,1) {};
\node[circ] (2) at (.8,2) {};
\node[circ] (3) at (2,1.4) {};
\node[circ] (4) at (1.8,.3) {};
\node[circ] (5) at (.8,0) {};

\draw (1) -- (2)
(2) -- (3)
(1) -- (4)
(2) -- (5)
; 

\end{tikzpicture}
\hspace{1cm}
\begin{tikzpicture}

\node[circ] (1) at (0,1) {};
\node[circ] (2) at (.8,2) {};
\node[circ] (3) at (2,1.4) {};
\node[circ] (4) at (1.8,.3) {};
\node[circ] (5) at (.8,0) {};

\draw (1) -- (2)
(2) -- (3)
(2) -- (4)
(2) -- (5)
(3) -- (4)
; 

\end{tikzpicture}
\hspace{1cm}
\begin{tikzpicture}

\node[circ] (1) at (0,1) {};
\node[circ] (2) at (.8,2) {};
\node[circ] (3) at (2,1.4) {};
\node[circ] (4) at (1.8,.3) {};
\node[circ] (5) at (.8,0) {};

\draw (1) -- (2)
(2) -- (3)
(2) -- (4)
(2) -- (5)
(3) -- (4)
(1) -- (4)
; 

\node (a) at (.3,.6) {\tiny $\{1, 2\}$};
\node (a) at (1.35,.1) {\tiny $\{1, 2, 3\}$};
\node (a) at (1.6,1.25) {\tiny $\{1, 3\}$};
\node (a) at (1.8,1.8) {\tiny $\{2, 3\}$};
\node (a) at (2.3,.7) {\tiny $\{1, 2\}$};
\node (a) at (.2,1.7) {\tiny $\{2, 3\}$};

\end{tikzpicture}
\caption{\CCIS instance with $G_1, G_2$ and $G_3$; and $\calG$ for the reformulation of \CCIS.}
\end{figure}

\begin{theorem}
\label{thm:ISWtwo}
There is no fixed-parameter algorithm for \textup{\textsc{ConsCheck: Independent}} \textup{\textsc{Set}[degree]} unless $\FPT=\W[2]$.
\end{theorem}

\begin{proof}
We present a reduction that given an instance $(\calU, \calF, k)$ of the classical \SCpr problem~\cite{DBLP:books/sp/CyganFKLMPPS15}, constructs an equivalent instance $\mathcal{I}$ of \CCIS.
Recall that an instance $(\calU, \calF, k)$ of \SCpr consists of a family $\calF$ of $m$ subsets over the $n$-element universe $\calU=[n]$, where we are interested in at most $k$ sets from $\calF$ whose union covers all of $\calU$.

Before proceeding with the reduction, we present a reformulation of \CCIS for the case where the instance $\mathcal{I}$ only consists of negative samples.

\vspace{0.2cm}\noindent\textit{Reformulation.}\quad
Consider an instance $\mathcal{I}$ of \CCIS that contains negative samples only, i.e. $\mathcal{I}=\{V, (E_i, 0)_{i\in [t]}\}$ for some $t$; to avoid overloading $k$, let us use $k'$ to denote the first component of the parameter for $\mathcal{I}$. 

Let us consider a graph $\calG$ over the vertex set $V$ and with the edge set $E(\mathcal{G})=\bigcup_{i=1}^{t} E(G_i)$. 
Now, for each edge $e\in E(\mathcal{G})$, we assign a label set $\mathcal{L}_e=\{i|e\in E(G_i)\}$. 
In other words, for each edge $e\in E(\calG)$ we remember the samples $e$ originates from.
  
Now, our aim is to select a set $H\subseteq V(\calG)$ of $k'$ vertices such that the union of all label-sets of edges between the vertices of $H$ is exactly $[t]$; formally, we seek a set $H$ of $k'$ vertices such that $\bigcup\limits_{v, u\in H}\calL_{uv}=[t]$.
Note that such a set $H\subseteq V(\calG)$ guarantees that, for each $i\in [t]$, there exist $u, v\in H$ such that $i\in \calL_{uv}$; the latter, by the definition of the label-set, means that $uv\in G_i$.
Thus, the set $H$ is not an independent set for any of the graphs $G_i=(V,E_i)$, $i\in [t]$.
Conversely, if for each $i\in [t]$ there are $u, v\in H$ such that $vu\in E_i$, then $uv\in \calL_{vu}$, and a union of the label-sets for each pair $u, v\in H$ then gives $[t]$.

\vspace{0.2cm}\noindent\textit{Construction.}\quad
Now, let us consider an instance $(U, \mathcal{F}, k)$ of \SCpr.
We construct an instance $(G, k')$ of the reformulated \CCIS problem described above, with $t=n$ being the size of the universe.
The graph $\calG$ is constructed as follows.
For each $F_i\in \mathcal{F}$, we introduce two vertices $f_i^1$, $f_i^2$, and add an edge $f_i$ with the set $F_i$ as its label-set $\calL_{f_i}$; $k'=2k$.

\vspace{0.2cm}\noindent\textit{Correctness.}\quad
Suppose, given an instance $(\calU, \calF, k)$ of \SCpr, that the reduction returns $(G, k')$ as an instance of the reformulated \CCIS.

Assume that $(\calU, \calF, k)$ admits $S\subseteq \calF$ such that both $\bigcup\limits_{F\in S}F=\calU$ and $|S|\leq k$ hold.
Then, let us construct a solution $H\subseteq V(G)$ for $(G, k')$.
For each $F\in S$, let us include both $f_i^1$ and $f_i^2$ to $H$.
Then, obviously, the condition on the set's cardinality will hold, i.e., $|H|\leq k'=2k$.
Also, $\bigcup\limits_{v, u\in H}\calL_{uv}=[n]$ since the label-sets are exactly the sets used in the solution $S$ for $(\calU, \calF, k)$ and $[n]$ is the universe.

For the other direction, suppose that we have $H\subseteq V(G)$ such that 
$\bigcup\limits_{v, u\in H}\calL_{uv}=[n]$.
Observe that $\calL_{uv}\neq \emptyset$ only for pairs of the form $f_i^1, f_i^2$ for some $i\in [m]$.
So, let $S\subseteq \calF$ be defined as follows: $S=\{F_i\mid \{f_i^1, f_i^2\}\subseteq H\}$.
Then, by the above construction, the sets in $S$ are identical to the label-sets of edges whose both endpoints are in $H$.
And, since the union of  the label-sets of edges whose both enpoints are in $H$ are $[n]$, the selected $S$ is a solution for the considered instance of \SCpr, i.e., $\bigcup\limits_{F\in S}F=\calU=[n]$.
\end{proof}

\begin{theorem}
\label{thm:DSFPT}
$\CCDS$ can be solved by a fixed-parameter algorithm running in time $\bigoh(2^{kd})\cdot |\mathcal{I}|$.
\end{theorem}

\begin{proof}
We show that \CCDS\ admits a kernel of size $k(d+1)$. 

Indeed, each vertex in the solution dominates at most $d$ vertices outside of the solution.
Thus, if $|V|>k\cdot (d+1)$ and \CCDS\ contains a positive sample, we can immediately output that there is no solution, while if if $|V|>k\cdot (d+1)$ and \CCDS\ contains negative samples only, we can correctly output an arbitrary set of $k$ vertices as a solution.

At this point, it remains to deal with instances of \CCDS\ such that $|V|\leq k\cdot (d+1)$. To deal with these, it suffices iterate over all subsets of $V$ and check for each such subset whether it is a solution for $\CCDS$.

Since each such check can be carried out in linear time, in $\bigoh(2^{kd})\cdot |\mathcal{I}
|$ we either find a solution for \CCDS or correctly identify that no such solution exists.
\end{proof}

Similarly to Theorems~\ref{thm:ColFPT},~\ref{thm:CCmatchfpt} and~\ref{thm:CCpathfpt}, we turn our attention to whether the lower bound for \CCIS\ can be overcome if the number of negative samples is bounded by the parameter. While the \W[2]-hardness  reduction of Theorem~\ref{thm:ISWtwo} does not hold if we are given a bound on the number of samples, it turns out that---unlike for \textsc{2-Coloring}, \textsc{Matching} and ($k$-)\textsc{Path}---consistency checking for \textsc{Independent Set} remains fixed-parameter intractable even under this additional restriction.

\begin{theorem}
\label{thm:ISWone}
There is no fixed-parameter algorithm for \textup{\textsc{ConsCheck: Independent}} \textup{\textsc{Set}[degree]}\xspace even when the number $t^-$ of negative saimples is assumed to be an additional parameter, unless $\FPT=\W[1]$.
\end{theorem}

\begin{proof}
We present a simple reduction that given an instance $(G, k')$ of the classical \textsc{Independent Set} decision problems on graphs, constructs an equivalent instance $\mathcal{I}$ of \CCIS\ where each graph $(V,E_i)$ has maximum degree $1$, $t^-=0$, and $k=k'$.

\vspace{0.2cm}\noindent\textit{Construction.}\quad
We construct the instance $\mathcal{I}=\{(V,(E_i, \lambda_i)_{i\in[t]}\}$ where $V$ is the set of vertices in the instance $(G,k')$ of \textsc{Independent Set}, and for each edge $e_j$ in $G$ the set $(E_i, \lambda_i)$ contains a tuple $(\{e_j\},1)$. Notice that $t$ is then the number of edges in $G$. 

\vspace{0.2cm}\noindent\textit{Correctness.}\quad
Suppose, given an instance $(G,k')$ of \textsc{Independent Set}, that the reduction
described above returns $\mathcal{I}=\{V, (E_i, \lambda_i)_{i\in [t]}\}$ as an instance of \CCIS. Then every independent set in $G$ is also an independent set for each of the positive samples in $\mathcal{I}$, and at the same time an independent set on $V$ that is independent for each of the positive samples in $\mathcal{I}$ is also an independent set in $G$. Hence, the existence of a fixed-parameter algorithm that solves every instance $\mathcal{I}$ obtained in this way parameterized by $k$ plus the number of negative samples plus the maximum degree of a sample would imply a fixed-parameter algorithm for \textsc{Independent Set}. 
\end{proof}

While restricting the number of negative samples alone is insufficient to achieve tractability, we conclude by showing that restricting the total number of samples allows for a fixed-parameter algorithm that solves the problem via a combination of multi-step exhaustive branching and color coding.

\begin{theorem}
\label{thm:CCISfpt}
\CCIS admits an algorithm which runs in time $(kdt)^{\bigoh(k^2)}\cdot n^{\bigoh(1)}$; in particular, it is fixed-parameter tractable when parameterized by $k+d+t$.
\end{theorem}

\begin{proof}
Consider an input instance $\mathcal{I}=\{V,(E_i, \lambda_i)_{i\in[t]}\}, k, d$ of \CCIS, and recall that the solution is a vertex set of size $k$. Let us denote the negative samples in $\mathcal{I}$ as $(E^-_1,0),\dots,(E^-_{t^-},0)$, where $t^-\leq t$ is the number of negative samples in the instance. The algorithm begins by exhaustively branching over all $k$-vertex graphs and all possible labelings of the edges of these graphs by subsets of $[t^-]$. We call each such graph considered in a separate branch in the algorithm a \emph{template}, and we note that the number of templates is upper-bounded by $2^{k^2}$.  We immediately discard templates such that there exists a label $z\in [t^-]$ which does not occur on any of its edges.

Intuitively, a template captures the behavior of a hypothetical solution with respect to the negative samples. Indeed, observe that for every hypothetical $k$-vertex solution $S$ of $\mathcal{I}$, we can construct a template $T_S$ as follows: 
\begin{itemize}
\item the vertices of $T_S$ are mapped to $S$ by an arbitrary bijection; 
\item whenever two vertices $s,t\in S$ are not adjacent to each other in any negative sample, we keep their counterparts non-adjacent in $T_S$; and
\item whenever two vertices $s,t\in S$ are adjacent to each other in negative samples $\{E_z~|~z\in Z\}$ for some $Z\subseteq [t^-]$, we place an edge between their counterparts in $T_S$ and label that edge by $Z$.
\end{itemize}

Second, in each branch where we have a fixed template $T$, we apply the color-coding technique with derandomization~\cite[Subsections 5.2 and 5.6]{DBLP:books/sp/CyganFKLMPPS15} to construct a family $\mathcal{B}$ of $k$-colorings of $V$ such that if there exists a solution $S=\{s_1,\dots,s_k\}$, then there will exist a coloring $B\in \mathcal{B}$ such that each $s_i\in S$ will receive the color $i$. We then branch over all colorings in $\mathcal{B}$, and the running time required for this branching step is upper-bounded by $(2e)^k\cdot k^{\bigoh(k)}\cdot n^{\bigoh(1)}$\cite[Subsection 5.6]{DBLP:books/sp/CyganFKLMPPS15}.

In the third step, the algorithm exploits the degree bound $d$ to enumerate all appropriately colored isomorphic copies of each of the connected components in $T$ which form independent sets in the positive instances. More precisely, for each connected component $C\in T$ consisting of vertices $s_{\alpha(1)},\dots,s_{\alpha(\ell)}$ for some $\ell\leq k$, the algorithm branches over all of the at most $|V|$ many choices of $s'_{\alpha(1)}\in V$ which received color $\alpha(1)$. It then chooses a new vertex in $C$ which is adjacent to $s_{\alpha(1)}$, say $s_{\alpha(i)}$, via an edge labeled by some edge-label set $Z$. To identify $s'_{\alpha(i)}\in V$, it chooses an arbitrary label $z\in Z$ and then branches over the at most $d$ many neighbors of $s'_{\alpha(1)}$ in $E^-_z$; for each such neighbor $v$, the algorithm tests whether $v$ is independent from $s'_{\alpha(1)}$ in all positive instances and whether $v$ is adjacent to $s'_{\alpha(1)}$ precisely in those negative instances whose indices are in $Z$. This branching procedure requires time at most $n\cdot d^{\ell-1}\leq n\cdot d^k$ and allows us to construct the set $L_C$ of all $\ell$-vertex subsets which are (1) independent in the positive samples, (2) colored in the same way as $C$, and (3) correspond to $C$ in the way described in the second paragraph of the proof.

In the final fourth step of the algorithm, assume we have constructed the sets $L_{C_1},\dots,L_{C_p}$ for the connected components $C_1,\dots, C_j$ of $T$. For each \emph{small} set $L_{C_j}$, $j\in [p]$, such that $|L_{C_j}|\leq \sum_{i\in [k]}{{kdt\cdot (dt)^k} \choose {i}}$, we perform exhaustive branching to choose a subset $S_j\in L_{C_j}$ and add it to a solution set $S'$. On the other hand, \emph{large} sets $L_{C_j}$  such that $|L_{C_j}|> \sum_{i\in [k]}{{kdt\cdot (dt)^k} \choose {i}}$ are ignored; the justification for this will become clear in the correctness argument, whereas the intuition is that in this case we are guaranteed to find a vertex subset in $L_{C_j}$ which will be independent from whichever other vertices are chosen to be part of the solution. Finally, we check whether the set $S'$ constructed in this branch is consistent with all positive samples; in particular, it is necessary to check that vertices originating from different components of $T$ are independent in all positive samples. If this test fails, the algorithm proceeds to the next branch. If $S'$ succeeds with this final test, the algorithm searches each large set $L_{C_j}$ until it finds an arbitrary set $S_j\in L_{C_j}$ which is independent from $S'$ in all positive samples, and adds $S_j$ to $S'$. It then outputs the constructed set $S'$. 

The running time of the algorithm described above is upper-bounded by $n^{\bigoh(1)}\cdot (kd^kt^k)^{\bigoh(k)}$. For correctness, let us assume the existence of a hypothetical solution $S$ and let $T_S$ be a template which corresponds to $S$ as described in the second paragraph of the proof. Consider the branch in which the algorithm considers the template $T_S$, and then a branch in which it selects a color family $B\in \mathcal{B}$ such that each vertex in $S$ receives a unique color. In the third step, the algorithm will construct the sets $L_{C_1},\dots,L_{C_p}$ for each of the $p$ connected components of $T_S$; in particular, there exists some $S_1\in L_{C_1},\dots,S_p\in L_{C_p}$ such that $S=\bigcup_{\ell\in [p]}S_\ell$. For the final branching step, let us consider the branch in which the algorithm correctly identifies those subsets $S_j$, $j\in [p]$, such that $S_j\subseteq S$ for each small $L_{C_j}$. 

To complete the proof, it remains to argue that the algorithm will extend the set $S_0\subseteq S$ constructed so far into some $k$-vertex solution for $\mathcal{I}$. To this end, notice that since $|S_0|\leq k$, there are at most $kdt$ many vertices that are adjacent to at least one vertex in $S_0$ in at least one \textbf{positive} sample; let us denote this vertex set $M$. Moreover, there are at most $kdt\cdot (dt)^k$ vertices in the distance-$k$ neighborhood of $M$ in the graph $(V,\bigcup_{u\in [t^-]}E^-_u)$ of all \textbf{negative} samples; let us denote the vertices in this distance-$k$ neighborhood by $M^+$. The total number of vertex subsets of size at most $k$ in $M^+$ is upper-bounded by $\sum_{i\in [k]}{{kdt\cdot (dt)^k} \choose {i}}$. This means that each large $L_{C_j}$ must contain at least one set, say $S^*_j$, which is not fully contained in $M^+$ and in particular contains at least one vertex outside of $M^+$. And since $S^*_j$ is connected and contains at most $k$ vertices, it must be completely distjoint from $M$. Hence $S^*_j$ must be non-adjacent to $S_0$. This guarantees that the algorithm will discover at least one set in the first large $L_{C_j}$ which can be added to its constructed set $S_0$. Crucially, since the only property of $S_0$ that was used in this argument was that $|S_0|\leq k$, it can be repeated in verbatim for every other large $L_{C_j}$. In summary the algorithm is guaranteed to output a set $S'\supseteq S_0$ which is a $k$-vertex solution for $\mathcal{I}$, as desired.

\end{proof}

\section{Concluding Remarks}
This article can be seen as a ``brief expedition into the forgotten island of consistency checking''---a place where \textsc{Split Graph} and \textsc{Edge Clique Cover} are tractable but \textsc{$2$-Coloring} and \textsc{Matching} are not, and where on bounded-degree graphs  \textsc{Independent Set} is \W[2]-hard while \textsc{Dominating Set} admits a fixed-parameter algorithm. 

To conclude on a more serious note, we remark that our understanding of parameterized consistency checking---and, more broadly, of sample complexity---is still in its infancy. Even in the setting of PAC learning considered here, we so far know very little about which learning problems belong to the classes \fptpac\ and \xppac. Still, we hope that the results and techniques presented in this article can contribute to bridging the gap between the parameterized (time) complexity and the sample complexity research fields. A natural target for future work in this direction would be to further deepen our understanding of problems such as learning CNF and DNF formulas~\cite{PittV88,AlekhnovichBFKP08,BrandGanianSimonov23} or juntas~\cite{MosselOS03}.

\bibliography{bib-bib}

\end{document}